\documentclass[11pt]{article}
\usepackage{fullpage}
\usepackage{amsmath}
\usepackage{amsthm}
\usepackage{authblk,thm-restate} 
\usepackage{amssymb, algorithm, algpseudocode,eqparbox,thmtools,enumerate,verbatim,bm,array,tabu}
\usepackage[usenames,dvipsnames,svgnames,table]{xcolor}
\definecolor{darkgreen}{rgb}{0.0,0,0.9}
\usepackage[colorlinks=true,pdfpagemode=UseNone,
citecolor=OliveGreen,linkcolor=BrickRed,urlcolor=BrickRed,
pdfstartview=FitH,pagebackref=true]{hyperref}
\usepackage{mathtools}

\renewcommand{\P}[1]{{\mathbb{P}}\left[#1\right]}

\newcommand{\E}[1]{{\mathbb{E}}\left[#1\right]}

\newcommand{\st}{\mbox{\rm s.t. }}
\providecommand{\norm}[1]{\left\lVert#1\right\rVert}

\def\equationautorefname~#1\null{%
  equation~(#1)\null
}

\newcommand{\threeops}[6]{
	\mathop{#1\vphantom{#3}\vphantom{#5}}\limits_{#2\vphantom{#4}\vphantom{#6}}
	\mathop{#3\vphantom{#5}\vphantom{#1}}\limits_{#4\vphantom{#6}\vphantom{#2}}
	\mathop{#5\vphantom{#1}\vphantom{#3}}\limits_{#6\vphantom{#2}\vphantom{#4}}
}

\declaretheorem[numberwithin=section]{theorem}
\declaretheorem[sibling=theorem]{lemma}

\declaretheorem[sibling=theorem]{claim}

\declaretheorem[sibling=theorem]{corollary}

\declaretheorem[sibling=theorem]{fact}
\declaretheorem[sibling=theorem]{definition}

\newenvironment{proofof}[1]{{\medbreak\noindent \em Proof of #1.  }}{\hfill\qed\medbreak}

\def\bone{{\bf 1}}

\def\eps{{\epsilon}}
\def\R{\mathbb{R}}
\def\C{\mathbb{C}}
\def\Z{\mathbb{Z}}

\def\bx{{\bf x}}

\def\cB{{\cal B}}
\def\cF{{\cal F}}

\def\cM{{\cal M}}

\def\cB{{\cal B}}

\def\by{{{\bf{y}}}}
\def\bx{{{\bf{x}}}}
\def\ba{{{\bf{a}}}}
\def\bb{{{\bf{b}}}}
\def\bz{{{\bf{z}}}}
\def\bc{{{\bf{c}}}}
\def\bw{{{\bf{w}}}}
\def\balpha{{{\bm{\alpha}}}}
\def\bbeta{{{\bm{\beta}}}}
\def\bdelta{{{\bm{\delta}}}}
\def\btalpha{{\bm{\tilde{\alpha}}}}
\def\bkappa{{{\bm{\kappa}}}}
\def\blambda{{{\bm{\lambda}}}}
\def\bty{{\bf{\tilde{y}}}}
\def\btz{{\bf{\tilde{z}}}}
\def\tz{\tilde{z}}

\newcommand{\complexity}[1]{\langle #1\rangle}

\DeclareMathOperator{\rank}{rank}

\DeclareMathOperator{\poly}{poly}

\DeclareMathOperator{\argmax}{argmax}

\DeclareMathOperator{\image}{Im}
\DeclareMathOperator{\per}{per}

\DeclareMathOperator{\conv}{conv}
\DeclareMathOperator{\OPT}{OPT}
\DeclareMathOperator{\newt}{newt}
\DeclareMathOperator{\supp}{supp}
\begin{document}

\title{A Generalization of Permanent Inequalities and\\ Applications in Counting and Optimization}

\author[1]{Nima Anari}
\author[2]{Shayan Oveis Gharan\thanks{This material is based on work supported by the National Science Foundation Career award.}}

\affil[1]{\small Stanford University\\ \texttt{anari@stanford.edu}}
\affil[2]{\small University of Washington\\ \texttt{shayan@cs.washington.edu}}

\date{}

\maketitle
\begin{abstract}
A polynomial $p\in\R[z_1,\dots,z_n]$ is real stable if it has no roots in the upper-half complex plane.	Gurvits's permanent inequality gives a lower bound on the coefficient of the $z_1z_2\dots z_n$ monomial of a real stable polynomial $p$ with nonnegative coefficients. This fundamental inequality has been used to attack several counting and optimization problems. 

	Here, we study a more general question: Given a stable multilinear polynomial $p$ with nonnegative coefficients and a set of monomials $S$, 
		we show that if the polynomial obtained by summing up all monomials in $S$ is real stable, 
	 then we can  lowerbound the sum of coefficients of monomials of $p$ that are in $S$. We also prove generalizations of this theorem to (real stable) polynomials that are not multilinear. We use our theorem to give a new proof of Schrijver's  inequality on the number of perfect matchings of a regular bipartite graph, generalize a recent result of Nikolov and Singh \cite{NS16}, and give deterministic polynomial time approximation algorithms for several counting problems. 
\end{abstract}
\thispagestyle{empty}
\newpage
\setcounter{page}{1}
\section{Introduction}
Suppose we are given a polynomial $p\in \R[z_1,\dots,z_n]$ with nonnegative coefficients.
For a vector $\bkappa\in \Z_+^n$ we write $C_p(\bkappa)$ to denote the coefficient of the monomial $\prod_{i=1}^n \bz^\bkappa$ in $p$.
Given a set $S$ of $\bkappa$'s, we want to estimate
$$ \sum_{\bkappa \in S} C_p(\bkappa).	$$
Let $\partial_{z_i}$ be the operator that performs partial differentiation  in $z_i$, i.e., $\partial/\partial z_i$.
We can rewrite the above quantity using a differential operator. 
Let 
$$ q(\partial\bz)=\sum_{\bkappa\in S} \prod_{i=1}^n \partial^{\kappa_i}_{z_i}$$
Then, 
$$\sum_{\kappa \in S} \bkappa! C_p(\bkappa)=q(\partial\bz)p(\bz)|_{\bz=0}.	
$$
More generally, given two polynomials $p,q\in \R[z_1,\dots,z_n]$ with nonnegative coefficients we analytically and algorithmically study the quantity
\begin{equation}\label{eq:pqcorrelation}
q(\partial\bz)p(\bz)|_{\bz=0} = p(\partial\bz)q(\bz)|_{\bz=0} = \sum_{\bkappa\in\Z_+^{n}} \bkappa!C_p(\bkappa)C_q(\bkappa).	
\end{equation}
for a large class of polynomials known as real stable polynomials.

Before stating our results, 
let us motivate this question. 

\paragraph{Permanent:} For a matrix $A\in\R_+^{n\times n}$ recall that the permanent of $A$ is
$$ \per(A)=\sum_{\sigma\in\mathbb{S}_n} \prod_{i=1}^n A_{i,\sigma(i)}.$$ 
Then, $\per(A)$ is $C_p(\bone)$ of the polynomial $p$,
$$ p(z_1,\dots,z_n)=\prod_{i=1}^n \sum_{j=1}^n A_{i,j} z_j.$$
\paragraph{Matroid Intersection:} Let $\cM_1([n],I_1)$ and $\cM_2([n],I_2)$ be two matroids on elements $[n]=\{1,2,\dots,n\}$ with set of bases $\cB_1,\cB_2$ respectively. Let
\begin{eqnarray*} 
p(\bz)&=&\sum_{B\in\cB_1} \prod_{i\in B} z_i,\\
q(\partial\bz)&=& \sum_{B\in\cB_2} \prod_{i\in B}\partial_{z_i}.
\end{eqnarray*}
Then, $p(\bz)q(\partial\bz)|_{\bz=0}$ is the number of bases in the intersection of $\cM_1$ and $\cM_2$.

We say a polynomial $p(z_1,\dots,z_n)\in\R[z_1,\dots,z_n]$ is {\em real stable} if it has no roots in the upper half complex plane.
Theory of stable polynomials recently had many applications in multiple areas of mathematics and computer science \cite{Gur06,MSS13,AO14,AO15}. Gurvits used this theory to give a new proof of the van der Waerden conjecture \cite{Gur06}. Most notably he proved that for any $n$-homogeneous real stable polynomial $p(z_1,\dots,z_n)$ with nonnegative coefficients and $q=\partial z_1\dots \partial z_n$,
\begin{equation}
\label{eq:gurvits}
e^{-n} \inf_{\bz>0} \frac{p(z_1,\dots,z_n)}{z_1\dots z_n}\leq q(\partial\bz)p(\bz)|_{\bz=0}\leq \inf_{\bz>0} \frac{p(z_1,\dots,z_n)}{z_1\dots z_n}.
\end{equation}
In other words, the coefficient of the monomial $z_1\dots z_n$ is at least $e^{-n} \inf_{z>0} \frac{p(z_1,\dots,z_n)}{z_1\dots z_n}$.
One can use this inequality to give a deterministic $e^{-n}$ approximation algorithm for the permanent of a nonnegative matrix.

In this paper we prove bounds analogous to \eqref{eq:gurvits} for the quantity $p(\bz)q(\partial\bz)$ when $p,q$ are real stable polynomials.
We discuss several applications in counting and optimization. We expect to see many more applications in the future.

\subsection{Our Contributions}

For two vectors $\bx,\by\in\R^n$ we write $\bx\by$ to denote the $n$ dimensional vector where the $i$-th coordinate is $x_iy_i$. For vectors $\bx,\by\in \R^n$ we write
$\bx^\by:=\prod_{i=1}^d x_i^{y_i}.$
We say a polynomial $p\in\R[z_1,\dots,z_n]$ is multilinear if $C_p(\bkappa)= 0$ for $\bkappa\notin\{0,1\}^n$. For a vector $\balpha\in \R^n$, we write $\balpha\geq 0$ to denote $\balpha\in \R_{+}^n$.


\begin{theorem}\label{thm:main}
	For any two real stable polynomials $p,q\in\R[z_1,\dots,z_n]$ with nonnegative coefficients,
	\begin{equation}\label{eq:main}
		\adjustlimits\sup_{\balpha\geq 0}  \inf_{\by,\bz>0} e^{-\balpha} \frac{p(\by)q(\bz)}{(\by\bz/\balpha)^\balpha} \leq q(\partial\bz)p(\bz)|_{\bz=0} \leq \adjustlimits\sup_{\balpha\geq 0}\inf_{\by,\bz>0}\frac{p(\by)q(\bz)}{(\by\bz/\balpha)^\balpha}.
	\end{equation}
	Furthermore, if $p,q$ are multilinear then 
	\begin{equation}\label{eq:multilinearmain}   \adjustlimits\sup_{\balpha\geq 0}   \inf_{\by,\bz>0} (1-\balpha)^{1-\balpha} \frac{p(\by)q(\bz)}{(\by\bz/\balpha)^\balpha} \leq q(\partial\bz)p(\bz)|_{\bz=0}.
	\end{equation}
\end{theorem}
Note that when $\|\balpha\|_\infty>1$, the term $(1-\balpha)^{1-\balpha}$ is ill-defined in the above theorem. However in that case we can replace it by any number, since by \autoref{fact:newton}, the LHS will be zero.

Before describing our main applications of the above theorem let us derive Gurvits's inequality \eqref{eq:gurvits}.
Let $q(\bz)=z_1z_2\dots z_n$. Then, it follows that the supremum is obtained by $\balpha=\bone$. For that $\balpha$ we have
$$ \inf_{\by,\bz>0} \frac{p(\by)q(\bz)}{(\by\bz/\balpha)^\balpha}=\inf_{\by>0} \frac{p(\by)}{y_1\dots y_n}.$$
In \autoref{sec:applications} we discuss several applications of the above theorem. The following two theorems are the main consequences that we  prove. 
Firstly, we give an algorithm  to approximate $q(\partial\bz)p(\bz)|_{\bz=0}$ for two given stable polynomials $p,q$ with nonnegative coefficients. We use this to give deterministic approximate counting algorithms.
In our results, we use convex programs involving certain polynomials. We make sure these convex programs can be approximately solved in time polynomial in the complexity of the involved polynomials. For a polynomial $p\in \R[z_1,\dots, z_n]$ with nonnegative coefficients, we define its complexity as $n+\deg p+|\log \min_{\bkappa:C_p(\bkappa)\neq 0} C_p(\bkappa)|+|\log \max_{{\bkappa:C_p(\bkappa)\neq 0}} C_p(\bkappa)|$, and we denote this by $\complexity{p}$.

\begin{restatable}{theorem}{counting}
	\label{thm:counting}
Given two  real stable polynomials $p,q\in\R[z_1,\dots,z_n]$ with nonnegative coefficients	
and an oracle that for any $\bz\in \R^n$, returns $p(\bz),q(\bz)$,
there is a polynomial time algorithm that  returns an $e^{\min\{\deg p,\deg q\}+\epsilon}$ approximation of  $q(\partial\bz)p(\bz)|_{\bz=0}$ that runs in time $\poly(\complexity{p}+\complexity{q}+\log(1/\epsilon))$.
\end{restatable}

Our next theorem gives a polynomial time algorithm to approximate $\max_{\bkappa\in\{0,1\}^n} C_p(\bkappa)C_q(\bkappa)$ for two given multilinear stable polynomials $p,q$ with nonnegative coefficients. We use this to generalize a recent results of Nikolov and Singh \cite{NS16}.
\begin{restatable}{theorem}{optimization}
\label{thm:optimization}
Given two real stable multilinear $D$-homogeneous polynomials $p,q\in\R[z_1,\dots,z_n]$ with nonnegative coefficients and an oracle that for any $\bz\in\R^n$ returns $p(\bz),q(\bz)$,
there is a polynomial time algorithm that  outputs an $e^{2D+\epsilon}$ approximation to 
$$ \max_{\bkappa\in\{0,1\}^n} C_p(\bkappa)C_q(\bkappa). $$
that runs in time $\poly(\complexity{p}+\complexity{q}+\log(1/\epsilon))$.
\end{restatable}

Independent of our work, Straszak and Vishnoi \cite{SV16} studied several variants of our problem. They show that if $p,q\in \R[z_1,\dots,z_n]$ are multilinear real stable polynomials with nonnegative coefficients, there is a convex relaxation to the quantity $q(\partial \bz) p(\bz)$ with an integrality gap  at most~$e^n$. In contrast, our approximation factor depends only on $\min\{\deg p, \deg q\}$ which can be significantly smaller than $n$.

\paragraph{Structure of the paper.}
In \autoref{sec:prelim} we describe many properties of real stable polynomials that we use throughout the paper. 
In \autoref{sec:applications} we describe several applications of \autoref{thm:main} in counting and optimization.
Finally, in sections \ref{sec:lowerbound} and \ref{sec:upperbound} we prove \autoref{thm:main}.
In particular, in \autoref{sec:lowerbound} we  lower bound $q(\partial\bz)p(\bz)|_{\bz=0}$, i.e., we prove the \eqref{eq:multilinearmain} and the LHS of \eqref{eq:main}.
And, in \autoref{sec:upperbound} we prove upper bound this quantity, i.e., the RHS of \eqref{eq:main}.

\section{Preliminaries}\label{sec:prelim}
Throughout the paper, we use bold letters to denote vectors.
For a vector $\by$, we write $\by\leq 1$ to denote that all coordinates of $\by$ are at most $1$.
For two vectors $\bx,\by\in \R^n$ we define $\bx\by=(x_1y_1,\dots,x_ny_n)$. Similarly, we define $\bx/\by=(x_1/y_1,\dots,x_n/y_n)$.
For a vector $\by\in\R^n$,  we define $\exp(\bx):=(e^{x_1},\dots,e^{x_n})$.

For vectors $\bx,\by\in \R^n$ we write 
$$\bx^\by:=\prod_{i=1}^d x_i^{y_i}.$$
As a special case, for a real number $c\in\R$ we write $c^x$ to denote $\prod_{i=1}^n c^{x_i}$. 

We use $\R_{++}=\{x: x>0\}$ to denote the set of positive real numbers. For an integer $n\geq 1$ we use $[n]$ to denote the set of numbers $\{1,2,\dots,n\}$. For any $m,n$, we let $\binom{[m]}{n}$ denote the collection of subsets of $[m]$ of size $n$.
Throughout the paper all logs are in base $e$.

For a vector $\bx\in\R^n$ and an integer $1\leq i\leq n$ we use $\bx_{-i}$ to denote the $m-1$ dimensional vector obtained from $\bx$ by dropping the $i$-th coordinate of $\bx$.
We use $\Z_+$ to denote the set of nonnegative integers.

We say a matrix $A\in\R^{n\times n}$ is {\em doubly stochastic} if all entries of $A$ are nonnegative and all row sums and columns sums are equal to 1,
\begin{eqnarray*} 
\sum_{j} A_{i,j} &=&1\ \forall i,\\
\sum_i A_{i,j} &=&1\  \forall j.
\end{eqnarray*}
\subsection{Stable Polynomials}\label{sec:stablepolynomials}
Stable polynomials are natural multivariate generalizations
of real-rooted univariate polynomials. For a complex number $z$, let
$\image(z)$ denote the imaginary part of $z$.
We say a polynomial $p(z_1,\dots,z_n)\in\C[z_1,\dots,z_n]$ is {\em stable}
if whenever $\image(z_i)>0$ for all $1\leq i\leq m$, $p(z_1,\dots,z_n)\neq 0$. As the only exception, we also call the zero polynomial stable. We say $p(.)$ is real stable, if it is stable and all of its coefficients are real. It is easy to see that any univariate polynomial is real stable  if and only if it is real rooted. 

We say a polynomial $p(z_1,\dots,z_n)$ is degree $k$-homogeneous, or $k$-homogeneous, if every monomial of $p$ has total degree exactly $k$. Equivalently, $p$ is $k$-homogeneous if for all $a\in\R$, we have
$$ p(a\cdot z_1,\dots,a\cdot z_n)=a^k p(z_1,\dots,z_n).$$
We say a monomial $z_1^{\alpha_1}\dots z_n^{\alpha_n}$ is {\em square-free} if $\alpha_1,\dots,\alpha_n\in\{0,1\}$. 
We say a polynomial is multilinear if all of its monomials are square-free.
For a polynomial $p$, we write $\deg p$ to denote the maximum degree of all monomials of $p$.
One of the famous examples of real stable polynomials is the set of elementary symmetric polynomials.
For an integer $k\geq 0$, the $k$-th elementary symmetric polynomial in $z_1,\dots,z_n$, $e_k(z_1,\dots,z_n)$ is defined as
$$ e_k(z_1,\dots,z_n) := \sum_{S\in \binom{[n]}{k}} \prod_{\ell \in S} z_i.$$
Note that the above polynomial is $k$-homogeneous and multilinear.

Let $\mu$ be a probability distribution on subsets of $[n]$. The {\em generating polynomial} of $\mu$ is defined as follows:
$$ g_\mu(z_1,\dots,z_n)=\sum_S \mu(S) \prod_{i\in S} z_i.$$
We say $\mu$ is a {\em strongly Rayleigh} (SR) distribution if $g_\mu$ is a real stable polynomial. Borcea, Br\"and\'en, and Liggett \cite{BBL09} defined (SR) distributions and proved numerous properties of them. 
Many interesting family of probability distributions are special cases of SR including product distributions, random spanning trees, and determinantal point processes. 
Note that $\mu(S)\geq 0$ for all $S\subseteq[n]$; so $g_\mu(S)$ has nonnegative coefficients. Many of theorems that we prove here have natural probabilistic interpretations if the underlying real stable polynomial is a generating polynomial of a SR distribution. 

We say a polynomial $p(y_1,\dots,y_n,z_1,\dots,z_n)$ is {\em bistable} if $p(y_1,\dots,y_n,-z_1,\dots,-z_n)$ is a real stable polynomial. 
\begin{fact}\label{fact:detstable}
	For any set of PSD matrices $A_1,A_2,\dots,A_n\in\R^{d\times d}$ and a symmetric matrix $B\in\R^{d\times d}$
	$$ p(z_1,\dots,z_n)\det(B+z_1A_1+\dots+z_nA_n)$$
	is real stable.
\end{fact}
\begin{fact}\label{fact:prodstable}
If $p(\bz)$ and $q(\bz)$ are real stable, then $p(\bz)q(\bz)$ is real stable.	
\end{fact}
\begin{fact}\label{fact:prodbistable}
	If $p(\by)$ and $q(\bz)$ are real stable, then $p(\by) q(\bz)$ is bistable.
\end{fact}
\begin{fact}\label{fact:substitutestable}
For any stable polynomial $p(z_1,\dots,z_n)$, an integer $1\leq i\leq n$, and any number $c$ with nonnegative imaginary part, $\image(c)\geq 0$, $p(z_1,\dots,z_{i-1},c,z_{i+1},\dots,z_n)$ is stable.
\end{fact}
\begin{fact}\label{fact:externalfield}
If $p(\by)$ is real stable, then for any nonnegative vector $\blambda\in\R_+^n$, $p(\blambda\by)$ is real stable. 
\end{fact}
\begin{lemma}\label{lem:1-xystability}
For any real stable polynomial $p(z_1,z_2,\dots,z_n)$ and any $1\leq i,j\leq n$, $ (1-\partial_{z_i}\partial_{z_j}) p$ is stable.
\end{lemma}
\begin{proof}
We say an operator $T$ is stability preserver if for any real stable polynomial $p(\bz)$, $Tp$ is real stable.
So, to prove the lemma we need to show that $1-\partial_{z_i}\partial_{z_j}$ is a stability preserver operator.
Borcea and Br\"and\'en \cite{BB10} showed that for any differential operator $T=\sum_{\bkappa} c_{\bkappa}\partial_z^\bkappa$, $T$ is stability preserver if the polynomial 
$$ \sum_{\bkappa\in \Z_+^n} c_{\bkappa}(-\bz)^\bkappa$$
is real stable. So, to prove the lemma it is enough to show that $1-z_iz_j$ is real stable. Obviously, if $\image(z_i),\image(z_j)>0$, then $z_iz_j\neq 1$. So, $1-z_iz_j$ is real stable. 
\end{proof}

Br\"and\'en \cite{Bra07} proved the following characterization of multilinear stable polynomials.
\begin{theorem}[{Br\"and\'en \cite[Theorem 5.6]{Bra07}}]\label{thm:stablenegcorrelation}
For any multilinear polynomial $p(y_1,\dots,y_m)$, $p$ is stable if and only if for all $1\leq i < j\leq m$,
$$ \partial_{y_i} p |_{y=x} \cdot \partial_{y_j}  p|_{y=x} \geq p(x) \cdot \partial_{y_i}\partial_{y_j} p|_{y=x},$$
for all $x\in\R^m$.
\end{theorem}

\subsection{Polarization} 
Let $p \in \C[z_1,\dots,z_n]$ be a stable polynomial of degree $d_i$ in $z_i$, $1 \leq i \leq n$. The polarization is an operator that turns $p$ into a  multilinear ``symmetric'' polynomial. 


For any $m\geq \max_i d_i$, The polarization $\pi_m(p)$ of $p$ is a polynomial in $m\cdot n$ variables $z_{i,j}$ for $1 \leq i \leq n, 1 \leq j \leq m$ defined as follows:
For all $1\leq i\leq n$, and $1\leq j\leq d_i$ substitute any occurrence of $z_i^k$ in $p$ with $\frac{1}{\binom{m}{k}} e_k(z_{i,1},\dots,z_{i,m})$. 

For example, if $p(z_1,z_2)=z_1^2+z_1z_2$, then
$$ \pi_2(p)=z_{1,1}z_{1,2} + \frac12(z_{1,1}+z_{1,2}) \cdot \frac12(z_{2,1}+z_{2,2}).$$

Note that for any $m$ and any polynomial $p$, $\pi_m(p)$ has the following properties:
\begin{itemize}
\item $\pi_m(p)$ is multilinear; 
\item $\pi_m(p)$ is symmetric in the variables $z_{i,1},\dots,z_{i,m}$ for all $1 \leq  i\leq  n$; 
\item if we substitute $z_{i,j} = z_i$ for all $i,j$ in $\pi_m(p)$ we recover $p$. 
\end{itemize}
Borcea, Br\"and\'en, and Liggett used Grace-Walsh-Szeg¨o's coincidence theorem \cite{Gra02,Sze22,Wal22}
to prove that stability is preserved under polarization \cite[Thm 4.7]{BBL09}.
\begin{theorem}[Borcea, Br\"and\'en, Liggett \cite{BBL09}]\label{thm:polarization}
If $p\in \C[z_1,\dots ,z_n]$ is stable then $\pi_m(p)$ is stable for all $m$.
\end{theorem}

\subsection{Stability and Convexity}
Polynomial optimization problems involving real stable polynomials with nonnegative coefficients can often be turned into concave/convex programs. Such polynomials are log-concave in the positive orthant:
\begin{theorem}[\cite{Gul97}]
\label{thm:guler}
For any real stable homogeneous polynomial $p\in\R[z_1,\dots,z_n]$ such that $p(\bz)>0$ for all $\bz\in\R_{++}^n$,  $\log p(\bz)$ is concave in the positive orthant.
\end{theorem}
\begin{proof}
Firstly, note that $\log p(\bz)$ is well-defined for any $\bz\in \R_{++}^n$.
It is enough to show that the function is concave along any interval in the positive orthant. Let $\ba\in\R_{++}^n$, $\bb\in\R^n$, and consider the line $\ba+t\bb$ where for any $t\in [0,1]$, $\ba+t\bb\in\R_{++}^n$. We show that $\log p(\ba+t\bb)$ is concave. Say $p$ is $k$-homogeneous, then
$$ p(\ba+t\bb)=p(t(\ba/t+\bb)) = t^k p(\ba/ t+\bb).$$
We claim that $p(\ba/t+\bb)$ is real rooted. 
Firstly, since $\ba\in\R_{++}^n$, and $p(.)$ is stable, $p(\ba t+\bb)$ is real rooted. But for any root $\lambda$ of $p(\ba t+\bb)$, $1/\lambda$ is a root of $p(\ba/t+\bb)$. Note that $p(\ba t+\bb)$ has no roots at $0$, because $\ba t+\bb\in\R_{++}^n$ for $t\in[0,1]$. So, $p(\ba/t+\bb)$ is real rooted.

Let $\lambda_1,\dots,\lambda_k$ be the roots of $p(\ba/t+\bb)$. We have
$$ p(\ba+t\bb)=t^k p(\ba/t+\bb)=t^k p(\ba) \prod_{i=1}^k (1/t - \lambda_i)=p(\ba)\prod_{i=1}^k (1-t\lambda_i).$$
We claim that for all $1\leq i\leq k$, $\lambda_i<1$. Otherwise, for some $t\in[0,1]$, $p(\ba+t\bb)=0$, but since $\ba+t\bb\in \R_{++}^n$, $p(\ba+t\bb)>0$ which is a contradiction. 
Therefore,
$$ \log p(\ba+t \bb)=\log p(\ba) + \sum_{i=1}^k \log (1-t\lambda_i).$$
The theorem follows by the fact that $\log(1-t\lambda)$ is a concave function of $t$ for $t\in [0,1]$ when $\lambda <1$.
\end{proof}
\begin{corollary}\label{cor:guler}
For any real stable polynomial $p\in\R[z_1,\dots,z_n]$ with nonnegative coefficients, $\log p(\bz)$ is concave in positive orthant. 	
\end{corollary}
\begin{proof}
The homogenization of a polynomial $p\in\R[z_1,\dots,z_n]$ is defined as the polynomial $p_H\in\R[z_1,\dots,z_n,z_{n+1}]$ given by
$$ p_H(z_1,\dots,z_{n+1})=z_{n+1}^{\deg p}p(z_1/z_{n+1},\dots,z_n/z_{n+1}).$$
Borcea, Br\"and\'en and Liggett \cite[Thm 4.5]{BBL09} show that for a polynomial $p$ with nonnegative coefficients, $p$ is real stable if and only $p_H$ is real stable.
Since $p_H$ is real stable, by \autoref{thm:guler}, $\log p_H(z_1,\dots,z_{n+1})$ is concave in the positive orthant. 
It follows that $p_H(z_1,\dots,z_n,1)$ is also concave in the positive orthant. But, $p(z_1,\dots,z_n)=p_H(z_1,\dots,z_n,1)$.  So, $p(\bz)$ is concave in the positive orthant.
\end{proof}

It is also an immediate corollary of H\"older's inequality that a polynomial with nonnegative coefficients is log-convex in terms of the log of its variables (for more details on log-convex functions see \cite{BV06}).
\begin{fact}
\label{fact:logconvex}
For any polynomial $p(y_1,\dots, y_n)$ with nonnegative coefficients, the expression \[\log p(\exp(\by))=\log p(e^{y_1},\dots,e^{y_n})\] is convex as a function of $\by$.
\end{fact}

\begin{lemma}
	\label{lem:alphanum}
	For any polynomial $p$ with nonnegative coefficients, and any $\balpha\geq 0$,
	\begin{equation*}
		\inf_{\bz>0}\frac{p(\bz)}{(\bz/\balpha)^\balpha}=\inf_{\bz>0}\frac{p(\balpha\bz)}{\bz^\balpha}.
	\end{equation*}
\end{lemma}
\begin{proof}
	If $\balpha>0$, the two sides become equal by a change of variable $\bz\leftrightarrow\balpha\bz$. If $\alpha_i=0$ for some $i$, since $p$ has nonnegative coefficients, it is obvious that one should let $z_i\to 0$ on the LHS, since the numerator decreases monotonically as $z_i\to 0$, and the denominator does not change. In the limit the LHS is the same as what we have on the RHS.
\end{proof}


\subsection{Jump Systems and Newton Polytope}
\begin{definition}
	For a polynomial $p(y_1,\dots,y_n)$, we define $\supp(p)$ as
	\[ \supp(p):=\{\bkappa\in\Z_+^n:C_p(\bkappa)\neq 0\}. \]
\end{definition}

Br\"and\'en \cite{Bra07} proved that if $p$ is a real stable polynomial then $\supp(p)$ forms what is known as a {\em jump system}.

\begin{definition}
	\label{def:jump}
	For two points $\bx, \by\in \Z^n$ define an $(\bx, \by)$-step to be any $\bdelta\in\Z^n$ such that $\|\bdelta\|_1=1$ and $\|\by-(\bx+\bdelta)\|_1=\|\by-\bx\|_1-1$. A set $\cF\subseteq \Z^n$ is called a jump system iff for every $\bx,\by\in \cF$ and $(\bx,\by)$-step $\bdelta$, either $\bx+\bdelta\in\cF$ or there exists an $(\bx+\bdelta,\by)$-step $\bdelta'$ such that $\bx+\bdelta+\bdelta'\in\cF$.
\end{definition}

Br\"and\'en \cite{Bra07} proved that for any real stable polynomial $p$, $\supp(p)$ is a jump system. We use the following immediate corollary  which follows from the fact that for any polynomial $p\in \R[y_1,\dots,y_n,z_1,\dots,z_n]$, $\supp(p(\by, \bz))=\supp(p(\by, -\bz))$.

\begin{fact}
	\label{fact:jump-system}
	If the polynomial $p(\by, \bz)$ is bistable, then $\supp(p)$ is a jump system.	
\end{fact}

Lov\'asz showed that the natural generalization of the greedy algorithm for maximizing linear functions over matroids, also works for jump systems \cite{L97}.
Given a finite jump system $\cF\subseteq \Z^n$, and $\bw\in\R^n$, the following algorithm finds $\max_{\bx\in \cF}\langle \bw, \bx\rangle$: Sort the coordinates in such a way that $|w_1|\geq |w_2|\geq \dots\geq |w_n|$. Fix the coordinates of $\bx$, one by one; for $i=1,\dots,n$, we will fix $x_i$ as follows: among all members of $\cF$ whose first $i-1$ coordinates match $x_1,\dots,x_{i-1}$, let $x_i$ be largest $i$-th coordinate if $w_i>0$ and let $x_i$ be the smallest otherwise.

\begin{theorem}[\cite{L97}]
	\label{thm:greedy}
	For any finite jump system $\cF\subseteq \Z^n$ and $\bw\in \R^n$, the above greedy algorithm correctly finds a point $\bx$ maximizing $\langle \bw, \bx\rangle$.
\end{theorem}

We will be working with the convex hull of $\supp(p)$, often known as the Newton polytope.

\begin{definition}
	For a polynomial $p(y_1,\dots,y_n)$, define the Newton polytope of $p$ as
	\begin{equation*}
		\newt(p):=\conv(\supp(p))=\conv\{\bkappa\in\Z_+^n:C_p(\bkappa)\neq 0\},
	\end{equation*}
	where $\conv$ denotes convex hull.
\end{definition}

We use the following fact in multiple places in the proofs.
\begin{fact}
	\label{fact:newton}
	For a polynomial $p(\by)\in \R [y_1,\dots,y_n]$ with nonnegative coefficients, and $\balpha\in \R_+^n$ we have
	\begin{equation*}
		\inf_{\by>0}\frac{p(\by)}{\by^\balpha}>0\iff \balpha\in \newt(p).
	\end{equation*}
\end{fact}
\begin{proof}
	Let $\balpha\in\newt(p)$; then we can write $\balpha$ as a convex combination of the vertices:  $\balpha=\sum_i \lambda_i\bkappa_i$, where $\lambda_i$'s are nonnegative and sum up to 1, and for all $i$,  $C_p(\bkappa_i)>0$. Therefore,
	\[
		p(\by)\geq \sum_i C_p(\bkappa_i)\by^{\bkappa_i}=\sum_i \lambda_i \frac{C_p(\bkappa_i)\by^{\bkappa_i}}{\lambda_i}\geq \prod_i \left(\frac{C_p(\bkappa_i)\by^{\bkappa_i}}{\lambda_i}\right)^{\lambda_i}=\by^\balpha\prod_i\left(\frac{C_p(\bkappa_i)}{\lambda_i}\right)^{\lambda_i},
	\]
	where for the last inequality we used weighted AM-GM. This implies that
	\[
		\inf_{\by>0}\frac{p(\by)}{\by^\balpha}\geq \prod_i\left(\frac{C_p(\bkappa_i)}{\lambda_i}\right)^{\lambda_i}>0.
	\]
	
	Conversely, suppose $\balpha\notin\newt(p)$; we show that $\inf_{\by>0}p(\by)/\by^\balpha=0$. Since $\newt(p)$ is convex, by separating hyperplane theorem, there exists $\bc\in\R^n$, such that $\langle \bc, \balpha\rangle >\langle \bc, \bkappa\rangle $ for every $\bkappa\in\newt(p)$. Now let $\by=\exp(t\bc)=(e^{tc_1},\dots,e^{tc_n})$, where $t>0$ will be fixed later. Then we have
	\[
		\frac{p(\by)}{\by^\balpha}=\frac{\sum_{\bkappa\in\Z_+^n}C_p(\bkappa)\by^\kappa}{\by^\balpha}\leq \exp\left(-t\left(\langle \bc, \balpha\rangle -\max_{\bkappa\in\newt(p))}\langle \bc, \bkappa\rangle\right)\right)\cdot\sum_{\bkappa\in\Z_+^n}C_p(\bkappa).
	\]
	It is easy to see that as $t\to\infty$, the RHS of the above inequality becomes zero.
\end{proof}

The convex hull of jump systems has been shown to be bisubmodular polyhedra in \cite{BC95}. The following is a restatement of this result.

\begin{theorem}[Restatement of \cite{BC95}]
	For a finite jump system $\cF\subseteq \Z^n$, define the function $f:\{-1,0,+1\}^n\to \R$ as
	\[ f(\bc):=\max_{\balpha\in\cF} \langle \bc, \balpha \rangle. \]
	Then
	\[ \conv(\cF) = \bigcap_{\bc\in\{-1,0,+1\}^n}\{\bx\in\R^n: \langle\bc, \bx\rangle\leq f(\bc) \}. \]
	Furthermore, $f$ is a bisubmodular function.
\end{theorem}
For the definition of bisubmodular functions, see \cite{BC95}. We will only use the following well-known fact about bisubmodular functions. For more details see \cite{MF08}.
\begin{fact}
	\label{fact:bisubmodular-separation}
	Suppose that $\cF$ is a jump system and $f$ is the bisubmodular function associated with $\conv(\cF)$. Given an oracle that returns $f(\bc)$ for any $\bc\in\{-1,0,+1\}^n$, we can construct a separation oracle for $\conv(\cF)$. Each call to the separation oracle uses at most $\poly(n)$ evaluations of $f$.
\end{fact}
\begin{proof}
	To construct a separation oracle, it is enough to design a procedure that given a point $\bx\in\R^n$, finds a $\bc\in\{-1,0,+1\}^n$ such that $\langle \bc, \bx \rangle > f(\bc)$ or declares that there is no such $\bc$.

	The function $\bc\mapsto f(\bc)-\langle \bc, \bx\rangle$ is bisubmodular, since $f$ is bisubmodular and $\bc\mapsto\langle \bc,\bx\rangle$ is modular. Bisubmodular functions can be minimized in strongly polynomial time, see e.g. \cite{MF08}. Therefore one can solve the following optimization problem using polynomially many evaluations of $f$.
	\[ \min_{\bc\in\{-1,0,+1\}^n}f(\bc)-\langle\bc,\bx\rangle. \]
	
	If the answer is nonnegative, then $x\in\conv(\cF)$. Otherwise, the optimizing $\bc$ produces the separating hyperplane.
\end{proof}


\section{Applications}
\label{sec:applications}
\subsection{Applications in Counting}
In this section we prove \autoref{thm:counting} and we discuss several applications of our results for counting problems.

We start with Schrijver's inequality. The following theorem is the main technical result of \cite{Sch98}.
\begin{theorem}[Schrijver \cite{Sch98}]
Let $A\in\R^{n\times n}$ be a doubly stochastic matrix, and let $\tilde{A}$ be defined as follows:
$$ \tilde{A}_{i,j} = A_{i,j} (1-A_{i,j})$$
Then,
$$ \per(\tilde{A}) \geq \prod_{i,j} (1-A_{i,j}).$$
\end{theorem}
Schrijver \cite{Sch98} used the above inequality to prove that any $d$-regular bipartite graph with $2n$ vertices has $\big(\frac{(d-1)^{d-1}}{d^{d-2}}\big)^n$ matchings. Gurvits and \cite{Gur11,GS14} used the above theorem to design a deterministic  $2^n$-approximation algorithm for the permanent.
We note that Gurvits used the machinery of stable polynomials to prove several lower bounds on permanent, but to the best of our knowledge his machinery could not capture  Schrijver's theorem.
Very recently, Lelarge \cite{Lel17} used the machinery developed by Csikv\'ari \cite{Csi14} to reprove Schrijver's theorem.

Here, we give a simple proof of this theorem using \autoref{thm:main}.
\begin{proof}
Let $\bz=\{z_{i,j}\}_{1\leq i,j\leq n}$ be a vector of $n^2$ variables. Define,
\begin{eqnarray*}
p(\bz) &=& \prod_{i=1}^n \sum_{j=1}^n A_{i,j} z_{i,j},\\
q(\bz) &=& \prod_{j=1}^n \sum_{i=1}^n (1-A_{i,j}) z_{i,j}.
\end{eqnarray*}
Then, it is not hard to see that
$$ q(\partial \bz) p(\bz) =\per(\tilde{A}).$$
This follows because each monomial of $p$ is a product $\prod_{i=1}^n z_{i,\sigma(i)}$ for some mapping $\sigma:[n]\to[n]$ and each monomial of $q$ is a product $\prod_{j=1}^n z_{\sigma(j),j}$ for some $\sigma:[n]\to[n]$. 

Let $\alpha_{i,j}=A_{i,j}$ for all $1\leq i,j\leq n$.
Since $p,q$ are multilinear, by \autoref{thm:main},
\begin{equation}\label{eq:pertA} \per(\tilde{A}) \geq (1-\balpha)^{1-\balpha} \inf_{\by,\bz>0} \frac{p(\by)q(\bz)}{(\by\bz/\balpha)^\balpha}.
\end{equation}
Firstly, observe that by weighted AM-GM inequality
$$ p(\by) = \prod_{i=1}^n \sum_{j=1}^n A_{i,j}y_{i,j} \geq \prod_{i,j} y_{i,j}^{A_{i,j}} = \by^\balpha,$$
where in the inequality we used that $A$ is  doubly stochastic.
Secondly, by another application of AM-GM inequality we have
$$ q(\bz) = \prod_{j=1}^n\sum_{i=1}^n A_{i,j} \frac{(1-A_{i,j})z_{i,j}}{A_{i,j}}\geq \prod_{i,j} \left(\frac{(1-A_{i,j})z_{i,j}}{A_{i,j}}\right)^{A_{i,j}}=((1-\balpha)\bz/\balpha)^\balpha,$$
where in the inequality we used that $A$ is doubly stochastic.
Plugging  the above two inequalities in \eqref{eq:pertA} we get
$$ \per(\tilde{A}) \geq (1-\balpha)^{1-\balpha} (1-\balpha)^\balpha = \prod_{i,j} (1-A_{i,j})$$
as desired.	
\end{proof}

\counting*
\begin{proof}
Firstly, we show that 
\begin{equation}\label{eq:Dapprox} 
e^{-\min\{\deg p,\deg q\}}\adjustlimits\sup_{\balpha\geq 0}   
\inf_{\by,\bz>0}  \frac{p(\by)q(\bz)}{(\by\bz/\balpha)^\balpha} 
\leq
q(\partial\bz)p(\bz)|_{\bz=0}
\leq
\adjustlimits\sup_{\balpha\geq 0}   
\inf_{\by,\bz>0}  \frac{p(\by)q(\bz)}{(\by\bz/\balpha)^\balpha} 
\end{equation}
Without loss of generality assume that $\deg p\leq \deg q$.
By \autoref{thm:main}, it is enough to show that for any $\balpha\geq 0$,
$$ e^{-\deg p} 
\inf_{\by,\bz>0}  \frac{p(\by)q(\bz)}{(\by\bz/\balpha)^\balpha} \leq 
\inf_{\by,\bz>0} e^{-\balpha} \frac{p(\by)q(\bz)}{(\by\bz/\balpha)^\balpha}.$$
To prove the above, it is enough to show that  if $\sum_i \alpha_i >\deg p$, then both sides of the above equal 0. Suppose $\langle \balpha,\bone\rangle >\deg p$. Let $\bz=\bone$ and $\by=t\bone$ for some parameter $t$. Then, as $t\to \infty$, $p(\by)/\by^\alpha\to 0$. This is because the degree of every monomial of $p(\by)$ is less than $\by^\balpha$. This was a special case of \autoref{fact:newton}. In general if $\balpha\notin\newt(p)\cap\newt(q)$, then both sides of the above are zero.

Using \eqref{eq:Dapprox}, to prove the theorem it is enough to design a polynomial time algorithm that computes 
$ \sup_{\balpha\geq 0} \inf_{\by,\bz>0}\frac{p(\by)q(\bz)}{(\by\bz/\balpha)^\balpha}.$
 Firstly, by \autoref{lem:alphanum} 
it is enough to solve $\sup_{\balpha\geq 0} \inf_{\by,\bz>0} \frac{p(\by)q(\balpha\bz)}{(\by\bz)^\balpha}.$
Now, let us do a change of variable $\by\leftrightarrow \exp(\by)$ and $\bz\leftrightarrow\exp(\bz)$.
By \autoref{fact:logconvex} and \autoref{cor:guler},
$$ \log \frac{p(\exp(\by))q(\balpha\exp(\bz))}{(\exp(\by)\exp(\bz))^\balpha}$$
is concave in $\balpha$ and convex in $\by,\bz$. 
Since the infimum of concave functions is also concave,
$\inf_{\by,\bz}\log \frac{p(\exp(\by))q(\balpha\exp(\bz))}{(\exp(\by)\exp(\bz))^\balpha}$ is a concave function of $\balpha$.
Therefore, 
$$ \adjustlimits\sup_{\balpha\geq 0}\inf_{\by,\bz>0} \frac{p(\exp(\by))q(\exp(\bz))}{(\exp(\by)\exp(\bz))^\balpha}.$$
is a convex program, or more precisely a concave-convex saddle point problem. In \autoref{sec:convex} we show how this convex program can be approximately solved in polynomial time. This gives an $e^{\deg p}$ approximation of $q(\partial\bz)p(\bz)|_{\bz=0}$ as desired. 
\end{proof}

Let us state some applications of the above result.
Our first application is in the context of recent results on Nash welfare maximization (NWM) problems. Recently the authors together with Saberi and Singh proved a generalization of Gurvits's inequality \eqref{eq:gurvits} to design an $e^n$ approximation for NWM problem \cite{AOSS17}.
Given a real stable polynomial $p\in\R[z_1,\dots,z_n]$ with nonnegative coefficients and an integer $1\leq D\leq n$, the main technical contribution of \cite{AOSS17} is a polynomial time algorithm that gives an $e^D$ approximation to 
$$ \Big(\sum_{S\in\binom{[n]}{D}} \prod_{i\in S} \partial_{z_i}\Big) p(\bz)|_{\bz=0}.$$
This result immediately follows from \autoref{thm:counting}.
Let $q$ be the $D$-th elementary symmetric polynomial on $z_1,\dots,z_n$, i.e., $q(\bz)=\sum_{S\in \binom{[n]}{D}} \prod_{i\in S} \partial_{z_i}$.
Recall that all elementary symmetric polynomials are real stable. Since $\deg q=D$,
we can approximate the above quantity up to an $e^D$ factor by \autoref{thm:counting}.

Our second application is a partial answer to an open problem of Kulesza and Taskar about determinantal point processes \cite{KT12}.
A determinantal point process (DPP) is a probability distribution $\mu$ on subsets of $[n]$  identified by a positive semidefinite ensemble matrix $L\in\R^{n\times n}$ where for any $S \subseteq [n]$ we have 
$$\P{S} \propto \det(L_{S,S}),$$ 
where $L_{S,S}$ is the square submatrix of $L$ corresponding to rows and columns indexed by $S$.
We refer interested readers to \cite{KT12} for applications of DPPs in machine learning. 
For an integer $0\leq k\leq n$, and a DPP $\mu$, the truncation of $\mu$ to $k$, $\mu_k$ is called a $k$-DPP, i.e.,
$$ \mu_k(S)\propto \begin{cases} \mu(S) & \text{if } |S|=k\\ 0 & \text{otherwise}.\end{cases}$$
Kulesza and Taskar left an open problem to estimate the quantity
$$ \sum_{S\subseteq[n]} \det(L_{S,S})^p$$
for some $p>1$. Their motivation is to study the Hellinger distance between two DPPs. 

Borcea, Br\"and\'en and Liggett \cite{BBL09} showed that DPPs and $k$-DPPs are strongly Rayleigh, i.e., if $\mu$ is a DPP, then $g_\mu(z_1,\dots,z_n)$ and $g_{\mu_k}(z_1,\dots,z_n)$ are real stable polynomials (see \autoref{sec:stablepolynomials}). 
Let $L,L'$ be ensemble matrices of two $k$-DPPs $\mu,\mu'$. Let $p=g_{\mu_k}$ and $q=g_{\mu'_k}$. Note that $p,q$ are $k$-homogeneous, multilinear, and have nonnegative coefficients. 
Therefore, by the above theorem we can estimate 
$$ q(\partial\bz)p(\bz)|_{\bz=0}=\sum_{S\subseteq [n]} \det(L_S)\det(L'_S)$$
up to a factor of $e^k$. To the best of our knowledge, to this date no approximation algorithm for the above quantity was known.

\subsection{Applications in Optimization}
In this section we prove \autoref{thm:optimization}. Let us restate the theorem.
\optimization*
\begin{proof}
We solve the following mathematical program and we show that its optimum solution gives an $e^{2D}$-approximation of the optimum. 
\begin{equation}
\begin{aligned}
	\adjustlimits\sup_{\balpha,\blambda\geq 0} \inf_{\by,\bz>0} &\frac{p(\blambda\by)q(\bz)}{(\by\bz/\balpha)^\balpha}\\
	\st \hspace{0.5cm}& \sum_i \lambda_i = D.
\end{aligned}\label{eq:cpoptimization}	
\end{equation}
	Firstly, we show that we can solve the above mathematical program. The idea is very similar to the proof \autoref{thm:counting}. We do a change of variable $\by\leftrightarrow\exp(\by)$, $\bz\leftrightarrow\exp(\bz)$ and we push the $\alpha$ in the denominator to $q$. So, instead we solve the following program:
$$
\begin{aligned}
	\adjustlimits\sup_{\balpha,\blambda\geq 0} \inf_{\by,\bz>0} \log\frac{p(\blambda\exp(\by)q(\balpha\exp(\bz))}{(\exp(\by)\exp(\bz))^\balpha}
\end{aligned}
$$
As before, $\log\frac{p(\blambda \exp(\by))q(\balpha\exp(\bz))}{(\exp(\by)\exp(\bz))^\balpha}$ is convex in $\by,\bz$ and concave in $\balpha,\blambda$. So, we can solve the above program in polynomial time; for more details see \autoref{sec:convex}. 
	
	Secondly, we show that the above program is a relaxation of the problem. 
	Let 
	$$\bkappa^*=\argmax_{\bkappa\in\{0,1\}^n} {C_p(\bkappa) C_q(\bkappa)}.$$ 
	Since every monomial of $p$ has degree exactly $D$, $\sum_i \kappa^*_i = D$, so $\blambda=\bkappa^*$ is a feasible solution. Let $\balpha=\bkappa^*$. 
	Then,
	$$ \inf_{\by>0} \frac{p(\blambda\by)}{\by^\balpha}\geq\inf_{\by>0}\frac{C_p(\bkappa^*)\by^{\bkappa^*}}{\by^{\bkappa^*}}=C_p(\bkappa^*)$$
	Similarly,
	$$ \inf_{\bz>0} \frac{q(\bz)}{(\bz/\balpha)^\balpha}=\inf_{\bz>0} \frac{q(\balpha \bz)}{\bz^\balpha} = \inf_{\bz>0} \frac{q(\bkappa^*\bz)}{\bz^{\bkappa^*}}=C_q(\bkappa^*),$$
	The first equality follows by \autoref{lem:alphanum}.
	Using the above two equations it follows that the optimum of \eqref{eq:cpoptimization} is at least $C_p(\bkappa^*)C_q(\bkappa^*)$.
	
	Next, we show that the optimum of \eqref{eq:cpoptimization} is within $e^{2D}$ factor of $C_p(\bkappa^*)C_q(\bkappa^*)$.
	Let $\lambda^*,\alpha^*$ be the arguments achieving optimum value in the $\sup$ in \eqref{eq:cpoptimization}\footnote{Note that since we can assume the domain of $\balpha,\blambda$ is compact, the supremum is attained at some point of the domain}.
	Consider the following probability distribution on $\bkappa\in\{0,1\}^n$. We sample $D$ numbers $X_1,X_2,\dots,X_D$ independently (with replacement) where for each $1\leq i\leq D$ we let $X_i=j$ with probability $\frac{\lambda^*_j}{D}$. We let $\kappa_i=1$, if $X_j=i$ for some $i$ and $\kappa_i=0$ otherwise. 
	We claim that
	\begin{equation} 
		\label{eq:lambdarounding}
		\E{C_p(\bkappa)C_q(\bkappa)}\geq e^{-2D} \OPT.
	\end{equation}
	Define
	$$ \tilde{p}(\bz)=p(\blambda^*\bz).$$
	Then, by \autoref{fact:externalfield}, $\tilde{p}$ is real stable. Therefore, by \autoref{thm:main}, 
	$$  \sum_{\bkappa\in\{0,1\}^n} C_{\tilde{p}}(\bkappa) C_q(\bkappa)\geq e^{-\balpha^*} \inf_{\by,\bz>0} \frac{\tilde{p}(\by)q(\bz)}{(\by\bz/\balpha^*)^{\balpha^*}}  \geq e^{-D} \adjustlimits\sup_{\balpha\geq 0} \inf_{\by,\bz>0} \frac{p(\blambda^*\by)q(\bz)}{(\by\bz/\balpha)^\balpha}\geq e^{-D}\OPT.$$
	The second to last inequality follows by \eqref{eq:Dapprox} and the last inequality follows by the fact that \eqref{eq:cpoptimization} is a relaxation of $\OPT$.
	
	So, to prove \eqref{eq:lambdarounding} it is enough to show that
	\begin{equation}\label{eq:tplambdarounding} \E{C_p(\bkappa)C_q(\bkappa)}\geq e^{-D} \sum_{\bkappa\in\{0,1\}^n} C_{\tilde{p}}(\bkappa)C_q(\bkappa).	
	\end{equation}
	Fix any monomial $\bkappa$ such that $C_p(\bkappa),C_q(\bkappa)>0$. Then, 
	$$ C_{\tilde{p}}(\bkappa)=C_p(\bkappa)\prod_{i:\kappa_i=1} \lambda^*_i.$$
	On the other hand,
	$$ \P{\bkappa}=D!\prod_{i:\kappa_i=1} \frac{\lambda^*_i}{D} \geq e^{-D} \prod_{i:\kappa_i=1} \lambda^*_i,$$
	The $D!$ factor comes from the fact that we can sample each such $\bkappa$ in $D!$ possible ways. 
	Plugging these two equations in \eqref{eq:tplambdarounding} proves it. This completes the proof of \autoref{thm:optimization}.
\end{proof}
One of the  consequences of the above theorem is a recent result of Nikolov and Singh \cite{NS16}. 
Given a PSD matrix $L$ and a partitioning $P_1,\dots,P_m$ of $[n]$, among all sets $S\subseteq [n]$ where for each $1\leq i\leq m$, $|S\cap P_i| = b_i$ we want to choose one that maximizes $\det(L_{S,S})$ where $L_{S,S}$ is the square submatrix of $L$ with rows and columns indexed by $S$. We refer interested readers to \cite{NS16} for applications of this problem. 

Let $k=\sum_i b_i$.
Let $v_1,\dots,v_n$ be the vectors of the Cholesky decomposition of $L$, i.e., for all $i,j$, $L_{i,j}=\langle v_i,v_j\rangle$. 
Let 
$$ p(y_1,\dots,y_n)=\partial_x^{n-k}\det(xI+y_1 v_1v_1^T+\dots+y_nv_nv_n^T)|_{x=0}.$$
By \autoref{fact:detstable} and \autoref{fact:substitutestable}   the above polynomial is real-stable. Furthermore, it is not hard to see that it is $k$-homogeneous and has nonnegative coefficients. In fact for any set $S\subseteq [n]$ of size $k$, the coefficient of the monomial $\bz^{\bone_S}$ is $\det(L_{S,S})$, where $\bone_S$ is the indicator vector of the set $S$. 
Now, let 
$$ q(z_1,\dots,z_n)=\prod_{i=1}^m e_{b_i}\left(\{z_j\}_{j\in P_i}\right)$$
be the generating polynomial corresponding to the partition matroid $P_1,\dots,P_k$. Again, by closure of stable polynomials under product, the above polynomial is real stable, multilinear, $k$-homogeneous, and has nonnegative coefficients. 
It follows by the above theorem that we can approximate
$$ \max_{S: |S\cap P_i|=b_i \forall i} \det(L_{S,S})$$
within an $e^{2k}$ factor. Note that our approximation factor is slightly worse but our theorem is significantly more general. In particular, we can replace $p,q$ in the above argument with any real stable multilinear polynomial. 

For example, suppose we are given a graph $G=(V,E)$ and we have assigned a vector $v_e$ to any edge $e$. We can let $q$ be the generating polynomial of the uniform spanning tree distribution of $G$. Then, using the above theorem we can approximate
$$ \max_T \det\Big(\sum_{e\in T} v_ev_e^T\Big)$$
within an $e^{2n}$ factor. 

\subsection{Solving the Convex Program}\label{sec:convex}

In this section we show how the convex programs from the previous two sections can be approximately solved in polynomial time. We measure the complexity of our algorithm in terms of the complexity of the involved polynomials. Recall the following definition.

\begin{definition}
	For a polynomial $p\in \R[z_1,\dots, z_n]$ with nonnegative coefficients, we define its complexity as $n+\deg p+|\log \min_{\bkappa:C_p(\bkappa)\neq 0} C_p(\bkappa)|+|\log \max_{{\bkappa:C_p(\bkappa)\neq 0}} C_p(\bkappa)|$, and we denote this by $\complexity{p}$.
\end{definition}

The main result of this section is the following.

\begin{theorem}
	Given a bistable polynomial $p\in \R[y_1,\dots,y_n,z_1,\dots,z_n]$ with nonnegative coefficients and an oracle that evaluates $p$ at any requested point $(\by,\bz)\in \R_+^{2n}$, one can find a $1+\epsilon$ approximation of
	\begin{equation}
		\label{eq:maincp}
		\adjustlimits\sup_{\balpha\geq 0}\inf_{\by,\bz>0} \frac{p(\by,\bz)}{(\by\bz/\balpha)^\balpha},
	\end{equation} 
	in time $\poly(\complexity{p}+\log(1/\epsilon))$.
\end{theorem}

In the rest of this section, we prove this theorem. We will use the ellipsoid method to compute the inner $\inf$ and also the outer $\sup$. The main difficulty for computing the $\inf$ is that the variables $\by, \bz$ are unbounded. In other words, we do not have a bounded outer ellipse for the ellipsoid method. In the first step, we use the fact that $\supp(p)$ is a jump system (\autoref{fact:jump-system}) to show that we can work with bounded $\by, \bz$ and only lose a small amount in the objective. The rest of the argument consists of constructing separation oracles and lower bounding the volume of inner ellipses for the ellipsoid methods.

\begin{lemma}
	\label{lem:bounded-domain}
	Given a polynomial $p\in\R[z_1,\dots,z_n]$ with nonnegative coefficients with the guarantee that $\supp(p)$ is a jump system and $\balpha\in \newt(p)$ and $\epsilon>0$, we have
	\[ \inf_{\bz:\norm{\log(\bz)}_\infty<M}\frac{p(\bz)}{\bz^\balpha}\leq(1+\epsilon)\inf_{\bz>0}\frac{p(\bz)}{\bz^\balpha}, \]
	for $M=\poly(\complexity{p}+\log(1/\epsilon))$.
\end{lemma}
\begin{proof}
	It is enough to prove that for any $\bz>0$, we can find $\bz^*$ with $\norm{\log(\bz^*)}_\infty<M$ such that \[\frac{p(\bz^*)}{(\bz^*)^\balpha}<(1+\epsilon)\frac{p(\bz)}{\bz^\balpha}.\] 
	For the convenience of notation, assume that $z_{n+1}=1$. Without loss of generality, after renaming coordinates, assume that \[|\log(z_1)|\geq |\log(z_2)|\geq \dots |\log(z_n)|\geq |\log(z_{n+1})|=0.\]
	
	Let $\Delta$ be a value bounded by $\poly(\complexity{p}+\log(1/\epsilon))$ that will be fixed later. Let $i$ be the coordinate for which $|\log(z_i)|-|\log(z_{i+1})|$ is maximized. If this maximum value is bounded by $\Delta$, we are done. Otherwise we change $\bz$ by subtracting a constant from  $|\log(z_1)|, \dots,|\log(z_i)|$ to make sure that $|\log(z_i)|-|\log(z_{i+1})|$ is reduced to $\Delta$ and $p(\bz)/\bz^\balpha$ does not grow by more than a $1+\epsilon/2n$ factor. We obtain $\bz^*$ by repeating this ``gap-closing'' procedure for no more than $n$ times.
	
	Define the vector $\bc\in\{-1,0,+1\}^n$ in the following way:
	\[ c_i :=\begin{cases}
		+1&\text{if $j\leq i$ and $z_j>1$},\\
		-1&\text{if $j\leq i$ and $z_j<1$},\\
		0&\text{otherwise}.
	\end{cases} \]
	
	Let $s=|\log(z_i)|-|\log(z_{i+1})|-\Delta$. Let $\btz=\exp(-s\bc)\bz$ (i.e., $\tz_i=e^{-sc_i}z_i$) be the result of the ``gap-closing'' operation.
	We just need to set $\Delta$ in such a way that
	\begin{equation}
		\label{eq:eps-grow}
		\frac{p(\btz)}{\btz^\balpha}\leq(1+\epsilon/2n)\frac{p(\bz)}{\bz^\balpha}. 
	\end{equation} 
	
	Let us partition $\cF=\supp(p)$ into $\cF_1\cup\cF_2$; let \[\cF_1:=\{\bkappa\in \cF:\langle \bc, \bkappa\rangle \geq \langle \bc, \balpha\rangle \},\] and let $\cF_2$ be the rest. Note that $\cF_1\neq \emptyset$ because $\balpha\in\newt(p)=\conv(\cF)$. 
	We can decompose $p=p_1+p_2$:
	\[ p(\bz)=\underbrace{\sum_{\bkappa\in \cF_1} C_p(\bkappa)\bz^\bkappa}_{p_1(\bz)}+\underbrace{\sum_{\bkappa\in \cF_2} C_p(\bkappa)\bz^\bkappa}_{p_2(\bz)}. \]
	To prove \eqref{eq:eps-grow}, it is enough to show the following two inequalities
	\begin{eqnarray}
		\label{eq:eps-grow-1}\frac{p_1(\btz)}{\btz^\balpha}&\leq & \frac{p_1(\bz)}{\bz^\balpha},\\
		\label{eq:eps-grow-2}p_2(\btz) & \leq & \frac{\epsilon}{2n} p_1(\btz).
	\end{eqnarray}
	
	First we prove \eqref{eq:eps-grow-1}. For any $\bkappa\in \cF_1$, we have
	\[
		\frac{\btz^\bkappa}{\btz^\balpha}=\btz^{\bkappa-\balpha}=(\exp(-s\bc)\bz)^{\bkappa-\balpha}=\exp(-s\langle \bc, \bkappa-\balpha\rangle)\frac{\bz^\bkappa}{\bz^\balpha}\leq \frac{\bz^\bkappa}{\bz^\balpha},
	\]
	where in the inequality we used $\bkappa\in\cF_1$ and $s\geq 0$. Multiplying the above inequality with $C_p(\bkappa)$ and summing over all $\bkappa\in\cF_1$ we get \eqref{eq:eps-grow-1}.
	
	It remains to prove \eqref{eq:eps-grow-2}. We will use the following claim to prove \eqref{eq:eps-grow-2}.
	
	\begin{claim}
		\label{claim:f1f2}
		For any $\bkappa\in\cF_2$, there exists $\bkappa'\in\cF_1$ such that
		\[ \btz^\bkappa \leq e^{-\Delta}\btz^{\bkappa'}. \]
	\end{claim}
	\begin{proof}
		Let $\bkappa'$ be the output of the greedy algorithm (with the ordering of coordinates being $1,\dots,n$) from \autoref{thm:greedy} for the jump system $\cF$ and linear function $\bx\mapsto \langle \bc, \bx\rangle$. We start from $\bkappa$ and move towards $\bkappa'$, ensuring that $\bkappa$ remains in $\cF$ and that $\btz^{\bkappa}$ never decreases and in at least one step it increases by a factor of $e^{\Delta}$.
		
		While $\bkappa\neq \bkappa'$, choose the smallest $j$ for which $\kappa_j\neq \kappa'_j$. If $j>i$, stop. Otherwise, by the greedy construction of $\bkappa'$, we know that $c_j(\kappa'_j-\kappa_j)> 0$. Let $\bdelta\in \Z^n$ be such that $\delta_j=c_j$ and $\delta_k=0$ for $k\neq j$. Note that $\bdelta$ is a $(\bkappa, \bkappa')$-step, see \autoref{def:jump}. Therefore either $\bkappa+\bdelta\in \cF$ or there exists a $(\bkappa+\bdelta, \bkappa')$-step $\delta'$ such that $\bkappa+\bdelta+\bdelta'\in \cF$. In the former case, we set $\bkappa$ to $\bkappa+\bdelta$ and in the latter case we set $\bkappa$ to $\bkappa+\bdelta+\bdelta'$. We show that after the move, $\btz^\bkappa$ does not decrease.
		
		\noindent {\bf Case 1 ($\bkappa+\bdelta\in\cF$).} Since $c_j \log(\tz_j)\geq c_i\log(\tz_i)\geq \Delta >0$, we get
		\[ \btz^{\bkappa+\bdelta}=z_j^{c_j}\btz^\bkappa \geq e^{\Delta}\btz^{\bkappa}\geq \btz^{\bkappa}. \]
		
		\noindent {\bf Case 2 ($\bkappa+\bdelta+\bdelta'\in \cF$).} Recall that $\delta'$ has only one nonzero coordinate and that coordinate is $\pm 1$. Let that be the $k$-th coordinate. Since the first $j-1$ coordinates of $\bkappa$ and $\bkappa'$ are the same, it must be that $k\geq j$. Therefore
		\[ \btz^{\bkappa+\bdelta+\bdelta'}=z_j^{c_j}z_{k}^{\delta'_k}\btz^\bkappa\geq z_j^{c_j}z_{k}^{-c_k}\btz^\bkappa\geq \btz^{\bkappa}. \]
		
		It remains to show that in one of the steps, $\btz^{\bkappa}$ grows by at least $e^{\Delta}$. Look at the quantity $\langle \bc, \bkappa\rangle$ and how it changes over the course of the above algorithm. At the beginning, since $\bkappa\in\cF_2$
		 \[\langle \bc, \bkappa \rangle < \langle \bc, \balpha\rangle. \]
		 When the algorithm finishes $\langle \bc, \bkappa \rangle=\langle \bc, \bkappa'\rangle$; this is because $\bkappa_j=\bkappa'_j$ for $j\leq i$ and $c_j=0$ for $j>i$. Note that $\langle \bc, \bkappa'\rangle \geq \langle \bc, \balpha\rangle$ because $\bkappa'\in \cF_1$. Therefore, at one point the quantity $\langle \bc, \bkappa\rangle$ must increase.  It is easy to see that when $\langle \bc, \bkappa\rangle$ increases, $\btz^\bkappa$ grows by at least a factor of $e^\Delta$. In case 1, $\btz^\bkappa$ always increases by a factor of $e^\Delta$. In case 2, $\langle \bc, \bkappa\rangle$ increases exactly when either $\delta'_k=c_k$ or $k>i$, and in both situations $\btz^\bkappa$ also increases by a factor of $e^\Delta$.
	\end{proof}
	
	Now let us finish the proof of \eqref{eq:eps-grow-2}. We have
	\[ p_2(\btz) = \sum_{\bkappa\in\cF_2}C_p(\bkappa)\btz^\bkappa\leq (\deg p+1)^n \cdot \max_{\bkappa\in \cF_2}C_p(\bkappa) \cdot \max_{\bkappa\in \cF_2}\btz^\bkappa ,\]
	where we used the fact that the number of terms in $p$ is bounded by $(\deg p+1)^n$. Similarly for $p_1$, we have
	\[ p_1(\btz) = \sum_{\bkappa\in\cF_1}C_p(\bkappa)\btz^\bkappa\geq \min_{\bkappa\in \cF_1}C_p(\bkappa)\cdot \max_{\bkappa\in\cF_1}\btz^\bkappa. \]
	Putting the above two inequalities and using \autoref{claim:f1f2}, we get
	\[ \frac{p_2(\btz)}{p_1(\btz)}\leq \frac{(\deg p+1)^n \cdot \max_{\bkappa\in \cF_2}C_p(\bkappa)}{\min_{\bkappa\in \cF_1} C_p(\bkappa)}\cdot\frac{\max_{\bkappa\in\cF_2}\btz^\bkappa}{\max_{\bkappa\in\cF_1}\btz^\bkappa}\leq e^{\poly(\complexity{p})-\Delta}. \]
	It is easy to see that by making $\Delta=\poly(\complexity{p})+\log(2n/\epsilon)=\poly(\complexity{p}+\log(1/\epsilon))$, we get the desired inequality \eqref{eq:eps-grow-2}.
\end{proof}

In the remaining part of this section, we first show that we can construct separation oracles, and then we bound the volume of the inner ellipses for the two ellipsoid methods. Let us construct the separation oracles for the two ellipsoid methods. For the ellipsoid computing $\inf_{\by,\bz>0}p(\by,\bz)/(\by\bz/\balpha)^\balpha$, the separation oracle just needs to compute the gradient of the ($\log$ of the) objective with respect to $\log(\by), \log(\bz)$. It is easy to see that this gradient can be computed from the partial derivatives of $p$ with respect to $\by, \bz$. The latter can be computed using the following lemma.

\begin{lemma}
	\label{lem:p-partial-compute}
	Given oracle access to a polynomial $p\in\R[y_1,\dots,y_n]$, for any $i$, the quantity $\partial_{y_i}p$ can be computed at any given point using $O(\deg p)$ oracle calls.
\end{lemma}
\begin{proof}
	Given a point $\by$, consider the polynomial 
	\[ q(t)=p(y_1,\dots,y_{i-1},y_i+t,y_{i+1},\dots,y_n). \]
	This is a univariate polynomial of degree at most $\deg p$. Therefore we can compute all of its coefficients by evaluating $p$ at $\deg p+1$ many points. Then we can simply output the coefficient of $t$.
\end{proof}

Now we construct the separation oracle for the ellipsoid method computing 
$$\adjustlimits\sup_{\balpha\geq 0} \inf_{\by,\bz>0}\frac{p(\by,\bz)}{(\by\bz/\balpha)^\balpha}.$$ 
Recall that if $\balpha\notin \newt(p)$, then $\inf_{\by,\bz>0}\frac{p(\by,\bz)}{(\by\bz/\balpha)^\balpha}=0$, see \autoref{fact:newton}. Given some $\balpha\geq 0$, we first check whether $\balpha\in \newt(p)$ using \autoref{fact:bisubmodular-separation}, and in case it is not, we return the promised separating hyperplane. So let us assume that $\alpha\in\newt(p)$.

In its current form, the expression $\frac{p(\by, \bz)}{(\by\bz/\balpha)^\balpha}$ is not concave in $\balpha$. So we need to use the same trick as in \autoref{lem:alphanum} to transform it into a concave form. We do this transformation \emph{after} finding  an approximate optimum for the $\inf$. We compute an approximate optimum of $\inf_{\by,\bz>0}\frac{p(\by,\bz)}{(\by\bz/\balpha)^\balpha}$ by restricting the domain of the variables $\by,\bz$ using \autoref{lem:bounded-domain} and running the inner ellipsoid method to desired accuracy.

Next, we find a separating hyperplane at $\alpha$, by calculating the gradient of
\[ \frac{p({\bbeta}(\by/\balpha), \bz)}{(\by\bz/\balpha)^{\bbeta}} \]
with respect to $\bbeta$ evaluated at $\bbeta=\balpha$. This can again be done using \autoref{lem:p-partial-compute}.

It remains to lower bound the volume of the inner ellipses in the ellipsoid methods. First consider the ellipsoid method  computing the $\inf$. It is enough to lower bound the volume of the following set for any given $\epsilon>0$, $\by^*,\bz^*>0$
\[ \left\{(\log(\by), \log(\bz))\mid \frac{p(\by,\bz)}{(\by\bz/\balpha)^\balpha}\leq (1+\epsilon) \frac{p(\by^*,\bz^*)}{(\by^*\bz^*/\balpha)^\balpha}\right\}. \]
To see this, it is enough to note that $\log(\frac{p(\by, \bz)}{(\by\bz/\balpha)^\balpha})$ is Lipschitz with respect to $\log(\by), \log(\bz)$ with Lipschitz constant bounded by $\poly(\deg p, n)=\poly(\complexity{p})$.

Now consider the ellipsoid method for computing the $\sup$. To bound the volume of the inner ellipse we are going to use a similar idea as before, except we prove a slightly weaker form of smoothness compared to being Lipschitz. We only consider the set of $\alpha$'s in $\newt(p)$. We are going to assume $\newt(p)$ is full-dimensional. If $\newt(p)$ is lower-dimensional we can use subspace-identification methods to restrict the ellipsoid method to the affine subspace spanned by $\newt(p)$, and proceed similarly to the full-dimensional case.

Let us start by considering the function 
\[g_{\by,\bz}(\balpha)= \log(\frac{p(\by,\bz)}{(\by\bz)^\balpha}). \]
It is easy to see that for any $\by,\bz>0$, $g_{\by,\bz}$ is Lipschitz with respect to $\balpha$ with Lipschitz constant bounded by $\poly(\log(\by), \log(\bz), n)$.
So the function 
\[g_M(\balpha)=\inf_{\by,\bz:\|\by\|_\infty,\|\bz\|_\infty<M}g_{\by,\bz}(\balpha) \]
is also Lipschitz with respect to $\balpha$ with Lipschitz constant bounded by $\poly(\log(M), n)$. Now consider the function $h(\balpha)=g_M(\balpha)+\balpha\log(\balpha)$. This function is not Lipschitz in $\balpha$, but it satisfies a weaker form of smoothness. In particular if $\alpha$ is perturbed by $\epsilon^2$, one can show that $h(\balpha)$ changes by at most $L\epsilon$ for some $L=\poly(\log(M), n)$. The function $h$ provides a uniform approximation of $\inf_{\by,\bz>0}\log(\frac{p(\by,\bz)}{(\by\bz/\balpha)^\balpha})$ over $\newt(p)$. Using this fact and that $\newt(p)$ has only integral vertices of bounded norm, one can construct the desired inner ellipse.

\section{The Lower Bound}\label{sec:lowerbound}
In this section we prove the LHS of \eqref{eq:main} and \eqref{eq:multilinearmain}. However, instead of lower bounding the quantity $\sum_{\bkappa\in\Z_+^n} C_p(\bkappa)C_q(\bkappa)$ for two real stable polynomials $p(\by),q(\bz)$, we treat $p(\by)\cdot q(\bz)$ as a new bistable polynomial on $2n$ variables $\by,\bz$. Recall that 
by \autoref{fact:prodbistable} the product of any two stable polynomials is a bistable polynomial.  This allows us to prove the lower bound using an inductive argument with a stronger statement. The following theorem implies the LHS of \eqref{eq:main}. 
\begin{theorem}\label{thm:gurgen}
For any  bistable polynomial  $p(y_1,\dots,y_n,z_1,\dots,z_n)$  with nonnegative coefficients and any vector $\alpha\in\R_+^n$,
\begin{equation}\label{eq:pyznonliear} \sum_{\bkappa\in \Z_+^n}\bkappa!\cdot C_p(\bkappa,\bkappa) \geq e^{-\balpha} \inf_{\by,\bz>0} \frac{p(\by,\bz)}{(\by\bz/\balpha)^\balpha}.
\end{equation}
\end{theorem}
Instead of directly proving the above theorem, we reduce the general case to the multilinear case. This reduction allows us to exploit Br\"and\'en characterization of real stable multilinear polynomials (see \autoref{thm:stablenegcorrelation}).
\begin{theorem}\label{thm:gurgenlinear}
For any bistable multilinear polynomial $p(y_1,\dots,y_n,z_1,\dots,z_n)$ with nonnegative coefficients and any vector $\alpha\in\R_+^n$,
$$ 
\sum_{\bkappa\in \{0,1\}^n} C_p(\bkappa,\bkappa) \geq  (1-\balpha)^{1-\balpha} \inf_{\by,\bz > 0} \frac{p(\by,\bz)}{(\by\bz/\balpha)^\balpha}$$
\end{theorem}
Note that not only do we use the above theorem to prove \autoref{thm:gurgen}, but we also use it to prove \eqref{eq:multilinearmain}. 
In the rest of this section, we prove \autoref{thm:gurgen} using the above theorem; then in \autoref{sec:lowerboundmultilinear} we prove the above theorem.
Observe that if $p$  is a multilinear bistable polynomial, then \autoref{thm:gurgen} follows from the above theorem (this is because $e^{-x} \leq (1-x)^{1-x}$ for any $0\leq x\leq 1$). So, we just need to reduce the case of a general $p$ to the multilinear case. 

Let $m$ be a very large integer that we fix later in the proof. For now, we assume $m$ is larger than the degree of all variables in $p$. Throughout the proof we use $\by,\bz$ to denote $n$-dimensional (positive) vectors and we use $\bty,\btz\in \R^{n\times m}$ to denote the variables of the polarization of $p$.
Let $q(\bty,\btz):=\pi_m(p)$ and $ r(\bty,\btz):=q(\bty, m\btz).$
Also, let $\btalpha\in \R_+^{n\times m}$ where for all $1\leq i\leq n$ and $1\leq j\leq m$,
$$ \tilde{\alpha}_{i,j} = \frac1m \alpha_i.$$
 By applying \autoref{thm:gurgenlinear} to $r(\bty,\btz)$ and $\btalpha$ we get 
\begin{equation}\label{eq:prlinearcon} \sum_{\bkappa\in\{0,1\}^{n\times m}} C_r(\bkappa,\bkappa) \geq (1-\btalpha)^{1-\btalpha} \inf_{\bty,\btz>0}\frac{r(\bty,\btz)}{(\bty\btz/\btalpha)^\btalpha}.	
\end{equation}
First, we show that the RHS of the above is at least the RHS of \eqref{eq:pyznonliear} for any $m$. Then, we show that as $m\to\infty$ the LHS of the above converges to the LHS of \eqref{eq:pyznonliear}.
\begin{lemma}\label{lem:prLHS}
	$$ (1-\btalpha)^{1-\btalpha} \inf_{\bty,\btz>0} \frac{r(\bty,\btz)}{(\bty\btz/\btalpha)^\btalpha} \geq e^{-\balpha}\inf_{\by,\bz>0} \frac{p(\by,\bz)}{(\by\bz/\balpha)^\balpha}.$$
\end{lemma}
\begin{proof}
Firstly, note that by the definition of $\btalpha$, $\btalpha^\btalpha = (\balpha/m)^\balpha$. Also, using the inequality $(1-x)^{1-x}\geq e^{-x}$, we have
$$ (1-\btalpha)^{1-\btalpha}\geq e^{-\btalpha}=e^{-\balpha}.$$
Therefore, to prove the lemma it is enough to show that
\begin{equation}
	\inf_{\bty,\btz>0} \frac{r(\bty,\btz)}{(m\bty\btz)^\btalpha} \geq \inf_{\by,\bz>0} \frac{p(\by,\bz)}{(\by\bz)^\balpha}. \label{eq:fromrtopbt}
\end{equation}

It follows by the change of variable $\btz\leftrightarrow\btz/m$ that
$$ \inf_{\bty,\btz>0} \frac{r(\bty,\btz)}{(m\bty\btz)^\btalpha} = \inf_{\bty,\btz>0} \frac{r(\bty,\btz/m)}{(\bty \btz)^\btalpha} = \inf_{\bty,\btz>0} \frac{q(\bty,\btz)}{(\bty\btz)^\btalpha},$$
where in the last identity we used $r(\bty,\btz/m)=q(\bty,m\btz/m)=q(\bty,\btz)$.
So, to prove \eqref{eq:fromrtopbt}, it is enough to show that 
\begin{equation}\inf_{\bty,\btz>0} \frac{q(\bty,\btz)}{(\bty\btz)^\btalpha} \geq \inf_{\by,\bz>0} \frac{p(\by,\bz)}{(\by\bz)^\balpha}. \label{eq:fromqtopbt}	
\end{equation}
We use the AM-GM inequality to prove this.
Fix an arbitrary $\bty,\btz\in \R^{n\times m}$, let $y_i=(\prod_j \tilde{y}_{i,j})^{1/m}$ and $z_i=(\prod_j \tilde{z}_{i,j})^{1/m}$ for all $i$.
We show that $\frac{q(\bty,\btz)}{(\bty\btz)^\balpha} \geq \frac{p(\by,\bz)}{(\by\bz)^\balpha}$ which proves \eqref{eq:fromqtopbt}. 

Firstly, $\bty^\btalpha=\by^\balpha$ and similarly $\btz^\btalpha=\bz^\balpha$. Therefore, all we need to show is that $p(\by,\bz)\leq q(\bty,\btz)$. Recall that $q=\pi_m(p)$; this means that we have substituted each term $y_i^k$ in $p$ with $\frac{1}{\binom{m}{k}}e_k(y_{i,1},\dots,y_{i,m})$. Since all coefficients of $p$ are nonnegative, to show $p(\by,\bz)\leq q(\bty,\btz)$ it is enough to show that for all $i,k,m$
$$ \frac{1}{\binom{m}{k}} e_k(\tilde{y}_{i,1},\dots,\tilde{y}_{i,m})\leq y_i^k,$$
or equivalently,
$$ \sum_{S\in \binom{[m]}{k}} \frac{1}{\binom{m}{k}}\prod_{j\in S} \tilde{y}_{i,j} \geq \left(\prod_{j=1}^m \tilde{y}_{i,j}\right)^{1/m}.$$
The above inequality is just an application of AM-GM. This proves \eqref{eq:fromqtopbt} and completes the proof of \autoref{lem:prLHS}.
\end{proof}
To finish the proof of \autoref{thm:gurgen} it remains to show that the LHS of \eqref{eq:prlinearcon} converges to the LHS of \eqref{eq:pyznonliear} as $m\to\infty$. We prove this statement term by term.

Fix $\bkappa\in\Z_+^n$ and consider the monomial $C_p(\bkappa,\bkappa) \prod_{i=1}^n y_i^{\kappa_i}z_i^{\kappa_i}$ of $p$. 
The contribution of this monomial to the LHS of \eqref{eq:pyznonliear} is $\bkappa!\cdot C_p(\bkappa,\bkappa).$
This monomial is substituted with
\begin{eqnarray*}  C_p(\bkappa,\bkappa) \prod_{i=1}^n \frac{e_{\kappa_i}(y_{i,1},\dots,y_{i,m})e_{\kappa_i}(mz_{i,1},\dots,mz_{i,m})}{\binom{m}{\kappa_i}^2} \\=C_p(\bkappa,\bkappa)\prod_{i=1}^n \frac{ m^{\kappa_i} e_{\kappa_i}(y_{i,1},\dots,y_{i,m})e_{\kappa_i}(z_{i,1},\dots,z_{i,m})}{\binom{m}{\kappa_i}^2}.
\end{eqnarray*}
in $r$. The contribution of this to the 
LHS of \eqref{eq:prlinearcon} is 
$$ C_p(\bkappa,\bkappa)\prod_{i=1}^n\frac{m^{\kappa_i}}{\binom{m}{\kappa_i}^2}\cdot \binom{m}{\kappa_i}=\bkappa!\cdot C_p(\bkappa,\bkappa)\cdot \prod_{i=1}^n \frac{m^{\kappa_i}}{m(m-1)\dots(m-\kappa_i+1)}.$$
As $m\to\infty$ the ratio in the RHS converges to $1$, so we get the same contribution.
This completes the proof of \autoref{thm:gurgen}.

\subsection{Multilinear Case}\label{sec:lowerboundmultilinear}
In this section we prove \autoref{thm:gurgenlinear}.
Our first observation is the following useful identity:
$$ \sum_{\bkappa\in\{0,1\}^n} C_p(\bkappa,\bkappa)=\prod_{i=1}^n (1+\partial_{y_i} \partial_{z_i}) p(\by,\bz) \Big|_{\by=\bz=0}.$$
To see this identity, note that any monomial of $p$ which is not of the form $\by^\bkappa \by^\bkappa$ for some $\bkappa$ is mapped to zero in the RHS. Furthermore, any monomial of the form $C_p(\bkappa,\bkappa)\by^\bkappa\bz^\bkappa$ is mapped to $C_p(\bkappa,\bkappa)$. So, throughout the proof we lower bound $\prod_{i=1}^n (1+\partial_{y_i} \partial_{z_i}) p(\by,\bz) \big|_{\by=\bz=0}$. The main reason that we use this reformulation is that it allows us to use induction. In the following lemma we show that for any bistable polynomial $p$, $(1+\partial_{y_i}\partial_{z_i})p$ is also bistable. So, we can apply the operator $\prod_{i=1}^n (1+\partial_{y_i}\partial_{z_i})$ inductively.

\begin{lemma}\label{lem:1+xy}
For any multilinear bistable polynomial $p(y_1,\dots,y_n,z_1,\dots,z_n)$ and any $1\leq i\leq n$, $(1+\partial_{y_i} \partial_{z_i})p|_{y_i=z_i=0}$ is also bistable.
\end{lemma}
\begin{proof}
Let $q(\by_{-i},\bz_{-i}) = (1+\partial_{y_i}\partial_{z_i}) p(\by,\bz)|_{y_i=z_i=0}$. We need to show that $q(\by_{-i},-\bz_{-i})$ is stable. We will prove that
\begin{equation}
\label{eq:1-xy=0operator}
q(\by_{-i},-\bz_{-i}) = (1-\partial_{y_i}\partial_{z_i}) p(\by,-\bz)|_{y_i=z_i=0}.
\end{equation}
Then, the claim follows from the fact the the polynomial on the RHS is stable. In particular,  since $p(\by,-\bz)$ is stable,
by  \autoref{lem:1-xystability}, $(1-\partial_{y_i}\partial_{z_i}) p(\by,-\bz)$ is stable. Furthermore,  by \autoref{fact:substitutestable} when we substitute $y_i=z_i=0$ in the latter we still have a stable polynomial.

It remains to show \eqref{eq:1-xy=0operator}. We can write $p$ as follows:
$$ p(\by,\bz) = y_i z_i p_1(\by_{-i},\bz_{-i}) + y_i p_2(\by_{-i},\bz_{-i}) + z_i p_3(\by_{-i},\bz_{-i}) + p_4(\by_{-i},\bz_{-i}). $$
So, 
\begin{eqnarray*} q(\by_{-i},-\bz_{-i}) &=& (1+\partial_{y_i}\partial_{z_i})p(\by,-z_1,\dots,-z_{i-1},z_i,-z_{i+1},\dots,z_n)  |_{y_i=z_i=0} \\
&=&  p_1(\by_{-i},-\bz_{-i}) + p_4(\by_{-i},-\bz_{-i})
\end{eqnarray*}
On the other hand,
\begin{eqnarray*}
(1-\partial_{y_i}\partial_{z_i}) p(\by,-\bz)|_{y_i=z_i=0} &=& (1-\partial_{y_i}\partial_{z_i}) \big(-y_iz_i p(\by_{-i},-\bz_{-i}) + y_ip_2(\by_{-i},-\bz_{-i})\\
&&~~
-z_ip_3(\by_{-i},-\bz{-i})+p_4(\by_{-i},-\bz_{-i})
\big) |_{y_i=z_i=0}\\
&=&  p_1(\by_{-i},-\bz_{-i}) + p_4(\by_{-i},-\bz_{-i})
\end{eqnarray*}
This proves \eqref{eq:1-xy=0operator}.
\end{proof}

We prove \autoref{thm:gurgenlinear} by induction.
See \autoref{lem:twovarbistable} for the base case of the induction.
Suppose the statement of the theorem holds for any bistable multilinear polynomial with nonnegative coefficients and at most $2n-2$ variables.
Given $p(y_1,\dots,y_n,z_1,\dots,z_n)$ and some $\eps>0$ we show that there exists $\by,\bz>0$ such that
\begin{equation}\label{eq:inductionstep}
\prod_{i=1}^n (1+\partial_{y_i} \partial_{z_i}) p \Big|_{\by=\bz=0}+\eps \geq (1-\balpha)^{1-\balpha} \frac{p(\by,\bz)}{(\by\bz/\balpha)^\balpha}.
\end{equation}

Let
$$ q(y_1,\dots,y_{n-1},z_1,\dots,z_{n-1})=(1+\partial_{y_n}\partial_{z_n}) p |_{y_n=z_n=0}.$$
Note that by definition, $q$ has nonnegative coefficients and is multilinear. 
Also, by \autoref{lem:1+xy}, $q$ is a bistable polynomial.  

Now, by the induction hypothesis, there exists $\by^*,\bz^*\in\R^{n-1}$ such that
$$ \prod_{i=1}^{n-1} (1+\partial_{y_i}\partial_{z_i}) q\Big|_{\by=\bz=0}+\eps/2 \geq (1-\balpha_{-n})^{1-\balpha_{-n}} \frac{q(\by^*,\bz^*)}{(\by^*\bz^*/\balpha_{-n})^{\balpha_{-n}}}.$$
Also, observe that by the definition of $q$, the LHS of the above is equal to the LHS of \eqref{eq:inductionstep}. So to prove \eqref{eq:inductionstep}, it is enough to show that there exist $y_n,z_n$ such that
$$
(1-\balpha_{-n})^{1-\balpha_{-n}} \frac{q(\by^*,\bz^*)}{(\by^*\bz^*/\balpha_{-n})^{\balpha_{-n}}}+ \eps/2\geq (1-\balpha)^{1-\balpha} \frac{p(\by^*,y_n,\bz^*,z_n)}{(\by^*\bz^*/\balpha_{-n})^{\balpha_{-n}} (y_nz_n/\alpha_n)^{\alpha_n}}. 
$$
Let $f(y_n,z_n)=p(\by^*,y_n,\bz^*,z_n)$. To prove the above it is enough to show that there exists $y_n,z_n$ such that
\begin{equation}\label{eq:inductionstepf}
(1+\partial_{y_n}\partial_{z_n}) f(y_n,z_n) |_{y_n=z_n=0} +\eps/2 \geq (1-\alpha_n)^{1-\alpha_n} \frac{f(y_n,z_n)}{(y_nz_n/\alpha_n)^{\alpha_n}}.
\end{equation}
We prove this using \autoref{lem:twovarbistable}. 
All we need to apply this lemma is that $f(y_n,z_n)$ is a multilinear bistable polynomial with nonnegative coefficients.
Since $p$ is multilinear with nonnegative coefficients, and we substitute $y_1,\dots,y_{n-1},z_1,\dots,z_{n-1}$ with positive numbers, we get that $f$ is multilinear with nonnegative coefficients.
Furthermore, since $p(y_1,\dots,y_n,-z_1,\dots,-z_n)$ is stable, by \autoref{fact:substitutestable}, 
$$f(y_n,-z_n)=p(y^*_1,\dots,y^*_{n-1},y_n,z^*_1,\dots,z^*_{n-1},-z_n)$$ is stable.
So, $f$ is a bistable polynomial. Therefore, \eqref{eq:inductionstepf} simply follows by \autoref{lem:twovarbistable}. This completes the proof of \autoref{thm:gurgenlinear}.

\begin{lemma}\label{lem:twovarbistable}
Let $p(y,z)$ be a multilinear bistable bivariate polynomial with nonnegative coefficients. Then, for any $\alpha\geq 0$,
$$ (1+\partial_y\partial_z) p\big|_{y=z=0} \geq (1-\alpha)^{1-\alpha} \inf_{y,z > 0}\frac{p(y,z)}{(yz/\alpha)^\alpha}$$
\end{lemma}
\begin{proof}
Since $p$ is multilinear,  if $\alpha>1$, then the RHS of the above is zero and we are done. In particular, for a fixed $z$, as $y\to\infty$, $p(y,z)/y^\alpha\to 0$.
So, throughout the proof we assume that $\alpha\leq 1$.

Since $p$ is multilinear we can write it as
$$ p(y,z) = ayz + by + cz + d,$$
for $a,b,c,d\geq 0$. 
To prove the lemma's conclusion, it is enough to show
\begin{equation}\label{eq:a+d} 
a+d \geq (1-\alpha)^{1-\alpha} \inf_{y,z > 0} \frac{ayz + by+cz +d}{(yz/\alpha)^\alpha}.
\end{equation}

When $\alpha=1$, the RHS of \eqref{eq:a+d} is at most $a$ since the ratio converges to $a$ as $y,z\to\infty$. When $\alpha=0$, the RHS of \eqref{eq:a+d} is at most $d$ since the ratio converges to $d$ as $y,z\to 0$. In both cases the RHS is at most $a+d$ and we are done. So for the rest of the proof assume that $0<\alpha<1$.

Since $p(x,-y)$ is stable, by \autoref{thm:stablenegcorrelation} we have
$$ -b\cdot c = \partial_y p(y,-z) |_{y=z=0} \cdot \partial_z p(y,-z) |_{y=z=0} \geq p(0,0) \cdot \partial_y\partial_z p(y,-z)|_{y=z=0} = -a\cdot d,$$
or, in other words $b\cdot c\leq  a\cdot d$. If this inequality is not tight, we can increase $b, c$ for it to become tight. In other words, let $b'\geq b$ and $c'\geq c$ be such that $b'c'=ad$. It is easy to see that such $b', c'$ can always be found. Then, to prove \eqref{eq:a+d}, it is enough to show that
\begin{equation*} 
a+d \geq (1-\alpha)^{1-\alpha} \inf_{y,z > 0} \frac{ayz + b'y+c'z +d}{(yz/\alpha)^\alpha},
\end{equation*}

Consider the following matrix:
\begin{equation*}
	M:=\left[\begin{array}{cc}a&b'\\c'&d\end{array}\right].
\end{equation*}

Since $\det(M)=ad-b'c'=0$, we have $\rank(M)\leq 1$, and in particular we can write $M$ as $uv^\intercal$ for $u,v\in\R^2$. Since the entries of $M$ are nonnegative we can further assume that $u,v\in \R_{+}^2$. Therefore we have
\begin{equation*}
	M = \left[\begin{array}{cc}a&b'\\c'&d\end{array}\right] = \left[\begin{array}{c}e\\ f\end{array}\right]\cdot\left[\begin{array}{cc}g&h\end{array}\right],
\end{equation*}
where $e,f,g,h\geq 0$. This implies the factorization $ayz+b'y+c'z+d=(ey+f)(gz+h)$. So to prove the conclusion of the lemma it is enough to show the following inequality:
\begin{equation}
	\label{eq:efgh}
	a+d=eg+fh\geq (1-\alpha)^{1-\alpha}\inf_{y, z>0}\frac{(ey+f)(gz+h)}{(yz/\alpha)^\alpha}=\alpha^\alpha(1-\alpha)^{1-\alpha} \inf_{y>0}\frac{ey+f}{y^\alpha}\cdot \inf_{z>0}\frac{gz+h}{z^\alpha}.
\end{equation}

To prove the above inequality we will use the following claim.
\begin{claim}
\label{claim:am-gm}
For any $0< \alpha< 1$ and $\beta, \gamma\geq 0$,
\begin{equation*}
	\inf_{t>0}\frac{\beta t+\gamma}{t^\alpha}=\frac{\beta^{\alpha}\gamma^{1-\alpha}}{\alpha^\alpha(1-\alpha)^{1-\alpha}}.
\end{equation*}
\end{claim}

Before proving the claim, let us use it to prove \eqref{eq:efgh}. Because of \autoref{claim:am-gm}, we can write the RHS of \eqref{eq:efgh} as
\begin{equation*}
	\alpha^{\alpha}(1-\alpha)^{1-\alpha}\cdot \frac{e^{\alpha}f^{1-\alpha}}{\alpha^\alpha(1-\alpha)^{1-\alpha}}\cdot \frac{g^{\alpha}h^{1-\alpha}}{\alpha^\alpha(1-\alpha)^{1-\alpha}}=\frac{(eg)^\alpha(fh)^{1-\alpha}}{\alpha^\alpha(1-\alpha)^{1-\alpha}}.
\end{equation*}

Applying \autoref{claim:am-gm} once more, with $\beta=eg, \gamma=fh$, we get

\begin{equation*}
	\frac{(eg)^\alpha(fh)^{1-\alpha}}{\alpha^\alpha(1-\alpha)^{1-\alpha}} =\inf_{t>0}\frac{egt+fh}{t^\alpha}\leq \left(\frac{egt+fh}{t^\alpha}\right)\Big|_{t=1}=eg+fh,
\end{equation*}
which is the LHS of \eqref{eq:efgh}.

It only remains to prove \autoref{claim:am-gm}.
\begin{proof}
	By the weighted AM-GM inequality we have
	\begin{equation*}
		\beta t+\gamma = \alpha\cdot\frac{\beta t}{\alpha}+(1-\alpha)\cdot\frac{\gamma}{1-\alpha}\geq \left(\frac{\beta t}{\alpha}\right)^\alpha\cdot \left(\frac{\gamma}{1-\alpha}\right)^{1-\alpha}=t^\alpha \frac{\beta^\alpha \gamma^{1-\alpha}}{\alpha^\alpha(1-\alpha)^{1-\alpha}}.
	\end{equation*}
	Therefore
	\begin{equation*}
	\inf_{t>0}\frac{\beta t+\gamma}{t^\alpha}\geq \frac{\beta^{\alpha}\gamma^{1-\alpha}}{\alpha^\alpha(1-\alpha)^{1-\alpha}}.
\end{equation*}
Note that we have equality in AM-GM when $\beta t/\alpha = \gamma/(1-\alpha)$. So it is enough to set $t=\frac{\alpha \gamma}{(1-\alpha)\beta}$ to get equality, which is possible if $\beta,\gamma>0$.

It remains to prove the claim when either $\beta=0$ or $\gamma=0$. In both cases the RHS is $0$ so we just need to show that the LHS is $0$ as well. When $\beta=0$, it is enough to let $t\to \infty$, and when $\gamma=0$, it is enough to let $t\to 0$ to get
\begin{equation*}
	\inf_{t>0}\frac{\beta t+\gamma}{t^\alpha}=0.
\end{equation*}
\end{proof}

This completes the proof of \autoref{lem:twovarbistable}.
\end{proof}

\section{The Upper Bound}
\label{sec:upperbound}
In this section we prove the RHS of \eqref{eq:main}. This will complete the proof of \autoref{thm:main}.

\begin{theorem}
	\label{thm:upperbound}
For any two real stable polynomials $p(y_1,\dots,y_n)$ and $q(z_1,\dots,z_n)$ with nonnegative coefficients,
	\begin{equation}
		\label{eq:upperbound}
		\sum_{\bkappa\in \Z_+^n}\bkappa^\bkappa \cdot C_p(\bkappa)C_q(\bkappa)\leq \adjustlimits\sup_{\balpha\geq 0}\inf_{\by,\bz>0}\frac{p(\by)q(\bz)}{(\by\bz/\balpha)^\balpha}.
	\end{equation}
\end{theorem}

Note that $\bkappa^\bkappa\geq \bkappa!$ for any $\bkappa$, so the above theorem immediately implies the RHS of \eqref{eq:main}.
Note that unlike \autoref{thm:gurgen} here we only prove our theorem for product of stable polynomials (as opposed to bistable polynomials). 
This is because the above statement fails for (multilinear) bistable polynomials. For example, the polynomial $p(y,z)=1+yz$ is bistable; however by \autoref{claim:am-gm}, for any fixed $0<\alpha<1$ we have
\begin{equation*}
	\inf_{y,z>0}\frac{p(y,z)}{(yz/\alpha)^\alpha}=\inf_{y,z>0}\frac{1+yz}{(yz/\alpha)^\alpha}=(1-\alpha)^{-(1-\alpha)}.
\end{equation*}
It is an easy exercise to see that $(1-\alpha)^{-(1-\alpha)}\leq e^{1/e}$. But the quantity $\sum_{\bkappa\in \Z_+^n}\bkappa^\bkappa\cdot C_p(\bkappa,\bkappa)$ is $2$ which is strictly more than $e^{1/e}$. 

%
%
%
In order to prove \autoref{thm:upperbound} we will crucially use the following minimax result about the RHS of \eqref{eq:upperbound}.
\begin{lemma}
	\label{lem:minimax}
	For any bistable polynomial $p(\by,\bz)$ with nonnegative coefficients,
	\begin{equation*}
		\adjustlimits\sup_{\balpha\geq 0}\inf_{\by,\bz>0}\frac{p(\by, \bz)}{(\by\bz/\balpha)^\balpha}=
		\threeops\inf{\by>0}\sup{\balpha\geq 0}\inf{\bz>0}\frac{p(\by, \bz)}{(\by\bz/\balpha)^\balpha}.
	\end{equation*}
\end{lemma}
Before proving \autoref{lem:minimax}, let us show how it can be used to prove \autoref{thm:upperbound}.

\begin{proofof}{\autoref{thm:upperbound}}
Because of \autoref{lem:minimax}, we have
	\begin{equation*}
		\adjustlimits\sup_{\balpha\geq 0}\inf_{\by,\bz>0}\frac{p(\by)q(\bz)}{(\by\bz/\balpha)^\balpha}=
		\threeops\inf{\by>0}\sup{\balpha\geq 0}\inf{\bz>0}\frac{p(\by)q(\bz)}{(\by\bz/\balpha)^\balpha}.
	\end{equation*}
 So to prove \eqref{eq:upperbound}, we just need to show that for any fixed $\by>0$, the following inequality holds:
\begin{equation*}
	\adjustlimits\sup_{\balpha\geq 0}\inf_{\bz>0}\frac{p(\by)q(\bz)}{(\by\bz/\balpha)^\balpha}\geq \sum_{\bkappa\in\Z_+^n}\bkappa^\bkappa\cdot C_p(\bkappa)\cdot C_q(\bkappa).
\end{equation*}
Note that $p(\by)$ is fixed on the LHS. So
\begin{eqnarray}
	\notag
	\adjustlimits\sup_{\balpha\geq 0}\inf_{\bz>0}\frac{p(\by)q(\bz)}{(\by\bz/\balpha)^\balpha}&=&p(\by)\adjustlimits\sup_{\balpha\geq 0}\inf_{\bz>0}\frac{q(\bz)}{(\by\bz/\balpha)^\balpha}\\
	\label{eq:sumkappa}
	&=&\sum_{\bkappa\in\Z_+^n}C_p(\bkappa)\by^\bkappa\adjustlimits\sup_{\balpha\geq 0}\inf_{\bz>0}\frac{q(\bz)}{(\by\bz/\balpha)^\balpha}.
\end{eqnarray}
But for any fixed $\bkappa\in\Z_+^n$, we have
\begin{eqnarray*}
	\adjustlimits\sup_{\balpha\geq 0}\inf_{\bz>0}\frac{q(\bz)}{(\by\bz/\balpha)^\balpha}&\geq& \left(\inf_{\bz>0}\frac{q(\bz)}{(\by\bz/\balpha)^\balpha}\right)\bigg|_{\balpha=\bkappa}=\inf_{\bz>0}\frac{q(\bz)}{(\by\bz/\bkappa)^\bkappa}\\
	& \geq & \inf_{\bz>0}\frac{C_q(\bkappa)\bz^\bkappa}{(\by\bz/\bkappa)^\bkappa}=\bkappa^\bkappa\cdot\frac{C_q(\bkappa)}{\by^\bkappa}.
\end{eqnarray*}
Plugging back into \eqref{eq:sumkappa} gives us
\begin{equation*}
	\adjustlimits\sup_{\balpha\geq 0}\inf_{\bz>0}\frac{p(\by)q(\bz)}{(\by\bz/\balpha)^\balpha}\geq \sum_{\bkappa\in\Z_+^n}C_p(\bkappa)\by^\bkappa\bkappa^\bkappa\frac{C_q(\bkappa)}{\by^\bkappa}=\sum_{\bkappa}\bkappa^\bkappa\cdot C_p(\bkappa)\cdot C_q(\bkappa).
\end{equation*}
This concludes the proof of \autoref{thm:upperbound}, assuming \autoref{lem:minimax}.
\end{proofof}

\subsection{Duality}

In this section we prove \autoref{lem:minimax}. From basic properties of $\inf$ and $\sup$, it is immediate that
\begin{equation*}
	\adjustlimits\sup_{\balpha\geq 0}\inf_{\by,\bz>0}\frac{p(\by,\bz)}{(\by\bz/\balpha)^\balpha}\leq
	\threeops\inf{\by>0}\sup{\balpha\geq 0}\inf{\bz>0}\frac{p(\by,\bz)}{(\by\bz/\balpha)^\balpha}.
\end{equation*}
So to prove equality, it is enough to prove that the LHS is greater than or equal to the RHS. We would ideally like to apply one of the classical minimax theorems, such as von Neumann's minimax theorem, which often require the function inside $\inf$ and $\sup$ to be convex-concave. However in its current form $p(\by,\bz)/(\by\bz/\balpha)^\balpha$ does not satisfy this property. Instead we transform both sides to a form suitable for minimax theorems.

Because of \autoref{fact:newton}, we can replace $\sup_{\balpha\geq 0}$ by $\inf_{\balpha\in[0,\deg p]^n}$. This is crucial for applying \autoref{thm:sion}, since the domain of $\balpha$ has to be compact.
Additionally, by \autoref{lem:alphanum}  we can move $\balpha$ from the denominator to the numerator, and inside of $p$. 
Finally, we apply the change of variables $\by\leftrightarrow\exp(\bty)$ and $\bz\leftrightarrow\exp(\btz)$. Therefore, to prove the above inequality, it is enough to show that
\begin{eqnarray}
	\label{eq:lhs1}
	\adjustlimits\sup_{\balpha\in [0,\deg p]^n}\inf_{\bty,\btz\in\R^n}\frac{p(\exp(\bty),\balpha\exp(\btz))}{\exp(\bty)^\balpha\exp(\btz)^\balpha} \geq \threeops\inf{\bty\in\R^n}\sup{\balpha\in[0,\deg p]^n}\inf{\btz\in\R^n}\frac{p(\exp(\bty),\balpha\exp(\btz))}{\exp(\bty)^\balpha\exp(\btz)^\balpha}.
\end{eqnarray}
	
We  use Sion's minimax theorem. We state a weaker form of this theorem below, which will be enough for our purposes; for the original stronger version, look at \cite{Sion58}.

\begin{theorem}[Immediate corollary of \cite{Sion58}]
	\label{thm:sion}
	Let $A$ be a compact convex subset of $\R^n$ and $B$ a (not necessarily compact) convex subset of $\R^m$. Suppose that $f:A\times B\to\R$ has the following properties:
	\begin{itemize}
		\item For each $a\in A$, the function $f(a,\cdot):B\to\R$ is continuous and quasi-convex.
		\item For each $b\in B$, the function $f(\cdot,b):A\to\R$ is continuous and quasi-concave.
	\end{itemize}
	Then we have
	\begin{equation*}
		\adjustlimits\sup_{a\in A}\inf_{b\in B}f(a,b)=\adjustlimits\inf_{b\in B}\sup_{a\in A}f(a,b).
	\end{equation*}
\end{theorem}
Recall that a real-valued function $f$ is quasi-convex if $f^{-1}(-\infty, t]$ is convex for any $t\in\R$ and $f$ is quasi-concave if $-f$ is quasi-convex.

We will apply \autoref{thm:sion} to the following function.
\begin{equation*}
	f(\balpha,(\bty,\btz)):=\frac{p(\exp(\bty),\balpha\exp(\btz))}{\exp(\bty)^\balpha\exp(\btz)^\balpha}.
\end{equation*}
For a fixed $\balpha\geq 0$, the function $f$ is log-convex, and therefore quasi-convex, in $\bty,\btz$ because of \autoref{fact:logconvex}. For a fixed $\bty,\btz$, the function $f$ is log-concave, and therefore quasi-concave, in $\balpha$ because of \autoref{cor:guler}. It is easy to see that $f$ is continuous in each variable as well. Therefore, by \autoref{thm:sion} we have
\begin{eqnarray*}
	\adjustlimits\sup_{\balpha\in [0,\deg p]^n}\inf_{\bty,\btz\in\R^n}\frac{p(\exp(\bty),\balpha\exp(\btz))}{\exp(\bty)^\balpha\exp(\btz)^\balpha}&=&\adjustlimits\inf_{\bty,\btz\in\R^n}\sup_{\balpha\in [0,\deg p]^n}\frac{p(\exp(\bty),\balpha\exp(\btz))}{\exp(\bty)^\balpha\exp(\btz)^\balpha}\\
	&\geq&\threeops\inf{\bty\in\R^n}\sup{\balpha\in[0,\deg p]^n}\inf{\btz\in\R^n}\frac{p(\exp(\bty),\balpha\exp(\btz))}{\exp(\bty)^\balpha\exp(\btz)^\balpha}.
\end{eqnarray*}
This proves \eqref{eq:lhs1} and finishes the proof of \autoref{lem:minimax}. 

\bibliographystyle{alpha}
\bibliography{ref2}

\begin{thebibliography}{AOSS17}

\bibitem[AO15a]{AO14}
Nima Anari and Shayan {Oveis Gharan}.
\newblock {Effective-Resistance-Reducing Flows, Spectrally Thin Trees, and
  Asymmetric TSP}.
\newblock In {\em FOCS}, pages 20--39, 2015.

\bibitem[AO15b]{AO15}
Nima Anari and Shayan {Oveis Gharan}.
\newblock The kadison-singer problem for strongly rayleigh measures and
  applications to asymmetric tsp.
\newblock 2015.

\bibitem[AOSS17]{AOSS17}
Nima Anari, Shayan {Oveis Gharan}, Amin Saberi, and Mohit Singh.
\newblock Nash social welfare, matrix permanent, and stable polynomials.
\newblock In {\em ITCS}, 2017.
\newblock to appear.

\bibitem[BB10]{BB10}
J.~Borcea and P.~Br\"ad\'en.
\newblock {Multivariate P\'olya-Schur classification problems in the Weyl
  algebra}.
\newblock {\em Proceedings of the London Mathematical Society}, 101(3):73--104,
  2010.

\bibitem[BBL09]{BBL09}
Julius Borcea, Petter Branden, and Thomas~M. Liggett.
\newblock {Negative dependence and the geometry of polynomials.}
\newblock {\em Journal of American Mathematical Society}, 22:521--567, 2009.

\bibitem[BC95]{BC95}
Andr{\'e} Bouchet and William~H. Cunningham.
\newblock Delta-matroids, jump systems, and bisubmodular polyhedra.
\newblock {\em SIAM J. Discrete Math.}, 8(1):17--32, 1995.

\bibitem[Br{\"a}07]{Bra07}
Petter Br{\"a}nd{\'e}n.
\newblock {Polynomials with the half-plane property and matroid theory}.
\newblock {\em Advances in Mathematics}, 216(1):302--320, 2007.

\bibitem[BV06]{BV06}
S.~Boyd and L.~Vandenberghe.
\newblock {\em {Convex Optimization}}.
\newblock Cambridge University Press, 2006.

\bibitem[Csi14]{Csi14}
P{\'e}ter Csikv{\'a}ri.
\newblock Lower matching conjecture, and a new proof of schrijver's and
  gurvits's theorems.
\newblock arXiv preprint arXiv:1406.0766, 2014.

\bibitem[Gra02]{Gra02}
JH~Grace.
\newblock The zeros of a polynomial.
\newblock In {\em Proc. Cambridge Philos. Soc}, volume~11, pages 352--357,
  1902.

\bibitem[GS14]{GS14}
Leonid Gurvits and Alex Samorodnitsky.
\newblock Bounds on the permanent and some applications.
\newblock In {\em FOCS}, pages 90--99. IEEE Computer Society, 2014.

\bibitem[G{\"u}l97]{Gul97}
Osman G{\"u}ler.
\newblock Hyperbolic polynomials and interior point methods for convex
  programming.
\newblock {\em MOR}, 22(2):350--77, 1997.

\bibitem[Gur06]{Gur06}
Leonid Gurvits.
\newblock Hyperbolic polynomials approach to van der waerden/schrijver-valiant
  like conjectures: Sharper bounds, simpler proofs and algorithmic
  applications.
\newblock In {\em STOC}, STOC '06, pages 417--426, 2006.

\bibitem[Gur11]{Gur11}
L.~Gurvits.
\newblock Unleashing the power of schrijver’s permanental inequality with the
  help of the bethe approximation.
\newblock arXiv:1106.2844, 2011.

\bibitem[KT12]{KT12}
Alex Kulesza and Ben Taskar.
\newblock Determinantal point processes for machine learning.
\newblock {\em arXiv preprint arXiv:1207.6083}, 2012.

\bibitem[Lel17]{Lel17}
Marc Lelarge.
\newblock Counting matchings in irregular bipartite graphs.
\newblock In {\em SODA}, 2017.
\newblock to appear.

\bibitem[Lov97]{L97}
L\'aszl\'o Lov\'asz.
\newblock The membership problem in jump systems.
\newblock {\em J. Combin. Theory Ser. B}, 70(1):45--66, 1997.

\bibitem[MF08]{MF08}
S.~Thomas McCormick and Satoru Fujishige.
\newblock Strongly polynomial and fully combinatorial algorithms for
  bisubmodular function minimization.
\newblock In {\em Proceedings of the {N}ineteenth {A}nnual {ACM}-{SIAM}
  {S}ymposium on {D}iscrete {A}lgorithms}, pages 44--53. ACM, New York, 2008.

\bibitem[MSS13]{MSS13}
Adam Marcus, Daniel~A Spielman, and Nikhil Srivastava.
\newblock {Interlacing Families II: Mixed Characteristic Polynomials and the
  Kadison-Singer Problem}.
\newblock 2013.

\bibitem[NS16]{NS16}
Aleksandar Nikolov and Mohit Singh.
\newblock Maximizing determinants under partition constraints.
\newblock In {\em Proceedings of the 48th Annual {ACM} {SIGACT} Symposium on
  Theory of Computing, {STOC} 2016, Cambridge, MA, USA, June 18-21, 2016},
  pages 192--201, 2016.

\bibitem[Sch98]{Sch98}
Alexander Schrijver.
\newblock Counting 1-factors in regular bipartite graphs.
\newblock {\em Journal of Combinatorial Theory, Series B}, 72(1):122--135,
  1998.

\bibitem[Sio58]{Sion58}
Maurice Sion.
\newblock On general minimax theorems.
\newblock {\em Pacific J. Math.}, 8:171--176, 1958.

\bibitem[SV16]{SV16}
Damian Straszak and Nisheeth~K. Vishnoi.
\newblock Real stable polynomials and matroids: Optimization and counting.
\newblock 2016.

\bibitem[Sze22]{Sze22}
G~Szeg{\"o}.
\newblock Bemerkungen zu einem satz von jh grace {\"u}ber die wurzeln
  algebraischer gleichungen.
\newblock {\em Mathematische Zeitschrift}, 13(1):28--55, 1922.

\bibitem[Wal22]{Wal22}
JL~Walsh.
\newblock On the location of the roots of certain types of polynomials.
\newblock {\em Transactions of the American Mathematical Society},
  24(3):163--180, 1922.

\end{thebibliography}
\end{document}